\RequirePackage{silence}

\documentclass[11pt]{article}
\usepackage[in]{fullpage}
\usepackage[normalem]{ulem}
\usepackage{parskip}
\usepackage{graphicx}

\usepackage{amsmath}
\usepackage{amsthm}
\usepackage{amssymb}
\usepackage{amsfonts}
\usepackage{bbm}
\usepackage{comment}
\usepackage{delarray}
\usepackage[T1]{fontenc}
\usepackage{keytheorems}
\usepackage{mathtools}
\usepackage{mfirstuc}
\usepackage[table]{xcolor}
\usepackage{xspace}

\usepackage{alglib-tomer}
\usepackage{thmlib-tomer}
\usepackage{levi-fonts}

\newcolumntype{Y}{>{\centering\arraybackslash}X}

\newcommand\accept[0]{\textsc{accept}\xspace}
\newcommand\reject[0]{\textsc{reject}\xspace}

\newcommand\abs[1]{\left|{#1}\right|}
\newcommand\cond{\middle |}
\newcommand\floor[1]{{\left\lfloor{#1}\right\rfloor}}
\newcommand\ceil[1]{{\left\lceil{#1}\right\rceil}}

\newcommand\poly{\mathrm{poly}}

\newcommand{\Ber}{\mathrm{Ber}}
\newcommand{\Bin}{\mathrm{Bin}}
\newcommand{\Geo}{\mathrm{Geo}}

\newcommand\dtv{\ensuremath{d_\mathrm{TV}}}

\newcommand\eps{\varepsilon}

\usepackage{hyperref}
\hypersetup{
	colorlinks,
	citecolor=blue,
	filecolor=black,
	linkcolor=black,
	urlcolor=black,
	linktoc=all
}

\newcommand\oracleDEG{\ensuremath{\mathrm{DEG}}}
\newcommand\oracleNEIGHBOR{\ensuremath{\mathrm{NEIGHBOR}}}
\newcommand\oracleIS{\ensuremath{\mathrm{IS}}}

\newcommand\paperZcoefZestZindZinv{\ensuremath{5}}

\title{Almost-Uniform Edge Sampling: Leveraging Independent-Set and Local Graph Queries}
\author{Tomer Adar\thanks{Technion - Israel Institute of Technology, Israel. Email: \href{mailto:tomer-adar@campus.technion.ac.il}{tomer-adar@campus.technion.ac.il}.} \and Amit Levi\thanks{University of Haifa, Israel. Email: \href{mailto:alevi@cs.haifa.ac.il}{alevi@cs.haifa.ac.il}.}}

\begin{document}

\begin{titlepage}
    \maketitle
    \thispagestyle{empty}
    \pagestyle{empty}
    
    \begin{abstract}
        A central theme in sublinear graph algorithms is the relationship between counting and sampling: can the ability to approximately count a combinatorial structure be leveraged to sample it nearly uniformly at essentially the same cost? 
        
        We study (i) independent-set (IS) queries, which return whether a vertex set $S$ is edge-free, and (ii) two standard local queries: degree and neighbor queries. Eden and Rosenbaum (SOSA ‘18) proved that in the local-query model, uniform edge sampling is no harder than approximate edge counting. We extend this phenomenon to new settings.

        We establish sampling-counting equivalence for the hybrid model that combines IS and local queries, matching the complexity of edge-count estimation achieved by Adar, Hotam and Levi (2026), and an analogous equivalence for IS queries, matching the complexity of edge-count estimation achieved by Chen, Levi and Waingarten (SODA `20).
        
        For each query model, we show lower bounds for uniform edge sampling that essentially coincide with the known bounds for approximate edge counting.
    \end{abstract}

    \newpage
    \nonumber
    \tableofcontents
\end{titlepage}

\section{Introduction}
A fundamental question in the field of sublinear-time graph algorithms concerns the relationship between \emph{approximate counting} and \emph{uniform sampling} (see e.g, \cite{er18,ERR19,AKK19,FGP20,EMR21,DLM22,AMM22,ENT23,ELRR25}). Can the ability to estimate the size of a combinatorial structure, such as the number of edges of a graph $G=(V,E)$, be leveraged to sample an element of that structure nearly uniformly at essentially the same cost?

The earliest investigation of this relationship focused on the case where the combinatorial structure is the set of edges $E$ of an unknown graph $G$. In this setting, the relationship was first explored within the \emph{local} query model. In this model, the algorithm is granted access to the graph's adjacency list, allowing it to perform \emph{degree queries} (returning the degree of a vertex $v$) and \emph{neighbor queries} (returning the $i$-th neighbor of a vertex $v$).

For a given $\eps\in(0,1)$, the tasks are formally defined as follows:
\begin{itemize}
    \item \textbf{Approximate counting:} With high probability the algorithm outputs a value $\tilde{\bm}$ satisfying $|E|(1-\eps)\le \tilde{\bm}\le (1+\eps)|E|$.
    \item \textbf{Almost uniform sampling:} With high probability the algorithm outputs $\be$ such that for every edge $e\in E$, ${(1-\eps)}/{|E|}\le\Pr[\be=e]\le ({1+\eps})/{|E|}$.
\end{itemize}

While early edge sampling methods often required significantly higher complexity than counting~\cite{KKR04}, a sequence of works~\cite{er18,TT22,ENT23} established that the query complexity of almost-uniform edge sampling is essentially equivalent to the query complexity of approximate edge counting (see~\cite{GR08}). 
Both tasks share query complexity of $\widetilde{\Theta}(n/\sqrt{ m})$, where $n$ is the number of nodes and $m$ is the number of edges in $G$ and the tilde notation suppresses polynomial factors in $\log n $ and $1/\eps$.

Beame, Har-Peled, Ramamoorthy, Rashtchian, and Sinha~\cite{beame20} introduced a notably different \emph{global} access by introducing the \emph{independent set} oracle (IS). In this model, an algorithm may query any arbitrary subset $S\subseteq V$, and the oracle responds with a single bit indicating whether $S$ induces an edge-free subgraph. 

Unlike local inspection, IS queries do not reveal the identity or specific location of edges. Instead, they provide coarse, global information about the presence or absence of edges within a queried set. This model can be viewed as a form of combinatorial group testing, where a single query probes a large portion of the graph simultaneously. As such, IS queries expose structural information that is fundamentally inaccessible through purely local inspection. The IS model and its variant received significant attention in the past years~\cite{BBGM19,BGMP19,BBGM21,DL21,AMM22,DLM22,DLM24} (see more information in Section~\ref{sec:intro:subsec:related-work}).

In this model, the complexity of approximate edge count was settled by Chen, Levi and Waingarten~\cite{clw20} at $\widetilde{\Theta}(\min\{\sqrt m, n/\sqrt{m}\})$. Most recently, Adar, Hotam and Levi~\cite{ahl26} considered a \emph{hybrid} model which integrates global IS queries with local queries. They showed that the query complexity of approximate edge count is essentially $\widetilde{\Theta}(\min\{\sqrt{m}, \sqrt{n/\sqrt{m}}\})$. Nonetheless, the corresponding sampling problems in both the IS and the hybrid models remained open.

\subsection{Our results}
In this paper, we bridge this gap by establishing that the sampling--counting equivalence extends to both the IS and Hybrid query models. In particular, we first present an efficient algorithm for edge sampling in the Hybrid model.
\begin{theorem}[note=Rephrase of Lemma \reflemma{sample-edge-hyb}, label=th:ubnd-hyb]
    There exists an algorithm $\proc*{sample-edge-hyb}(G,\eps)$ where $G$ is a graph over $n$ vertices (known to the algorithm) and $m$ edges (unknown to the algorithm) and $\eps > 0$, that makes $O(R \log n \log (n/\eps))$ degree, neighbor and independent-set queries, where $R = \min\{\sqrt{m}, \sqrt{n/\sqrt{m}}\}$, whose random output $\bX$ satisfies:
    \begin{itemize}
        \item $\Pr\left[\bX \ne \reject\right] \ge 2/3$.
        \item For every edge $e$ in $G$, $\Pr\left[\bX = e\; \cond\; \bX \ne \reject\right] \in (1 \pm \eps) / m$.
    \end{itemize}
\end{theorem}

We complement this algorithmic result by showing that the query complexity of \proc{sample-edge-hyb} is optimal within the hybrid model, up to logarithmic factors.

\begin{theorem}[note=Rephrase of Lemma \reflemma{lbnd-sample-edge-hyb}]
    Any algorithm whose behavior matches the guarantees of Theorem \ref{th:ubnd-hyb} must make $\Omega(\min\{\sqrt{m}, \sqrt{n / \sqrt{m}}\})$ queries in expectation.
\end{theorem}

Next, we shift our focus to the more restricted Independent Set query model.

\begin{theorem}[note=Rephrase of Lemma \reflemma{sample-edge-IS}, label=th:ubnd-IS]
There exists an algorithm $\proc*{sample-edge-IS}(G,\eps)$ where $G$ is a graph over $n$ vertices (known to the algorithm) and $m$ edges (unknown to the algorithm) and $\eps > 0$, that makes $O(R \cdot \poly(\log (n/\eps)))$ independent-set queries, where $R = \min\{\sqrt{m}, n/\sqrt{m}\}$, whose random output $\bX$ satisfies:
    \begin{itemize}
        \item $\Pr\left[\bX \ne \reject\right] \ge 2/3$.
        \item For every edge $e$ in $G$, $\Pr\left[\bX = e \;\cond\; \bX \ne \reject\right] \in (1 \pm \eps) / m$.
    \end{itemize}
\end{theorem}

Finally, we establish a corresponding lower bound for the IS model, confirming that the dependence on $m$ and $n$ in Theorem \ref{th:ubnd-IS} is essentially tight.
\begin{theorem}[note=Rephrase of Lemma \reflemma{lbnd-sample-edge-IS}]
    Any algorithm whose behavior matches the guarantees of Theorem \ref{th:ubnd-IS} must make $\Omega(\min\{\sqrt{m}, n / \sqrt{m}\})$ queries in expectation.
\end{theorem}

\subsection{Related work}
\label{sec:intro:subsec:related-work}

Since the initial work in the local model, the problem of almost-uniform edge sampling has been extended to various settings. Eden, Ron, and Rosenbaum~\cite{ERR19} showed that if a graph has bounded arboricity\footnote{Note that since $\alpha\le \sqrt{m}$ for every $G$ with $m$ edges, this result is at least as good as the result in \cite{er18,TT22,ENT23}} $\alpha$, the query complexity for almost-uniform sampling can be improved to $\widetilde{O}(\alpha n/m)$. 
Eden, Mossel and Rubinfeld~\cite{EMR21} gave an algorithm which samples $t$ edges using $\widetilde{O}(\sqrt{t}\cdot n/\sqrt{m})$ queries --  a notable improvement compared to naively using $t$ invocation of the edge sampler. This query complexity was later proved to be optimal by~\cite{TT22}. 

 Beyond the specific focus on edge estimation, several recent works have explored the broader problem of counting and sampling arbitrary subgraphs $H$ within a variety of expanded query models. \cite{ERR22} extended the results of~\cite{ERR19} for sampling $k$-cliques in the local model. In particular, their results show that for $k$-cliques sampling is harder than counting. In contrast, \cite{ELRR25} showed that for Hamiltonian subgraphs the complexity of counting and sampling coincides. 
 
Recent research has expanded the investigation of the sampling–counting relationship to arbitrary subgraphs $H$, utilizing various combinations of degree, neighbor, and sampling oracles to establish complexity bounds based on structural parameters like fractional edge covers and decomposition costs \cite{ABG+18, AKK19, FGP20,BER21}. Notably, these works emphasize that the cost of uniform subgraph sampling can often be made to match the complexity of approximate counting, even as the target structures become more complex.

The introduction of the independent-set oracle has spurred a broader investigation into global query primitives. \cite{beame20} introduced the Bipartite Independent Set (BIS) oracle, which queries whether an edge exists between two disjoint vertex subsets, and used it to provide bounds for edge estimation. These bounds were later strengthen by~\cite{DLM22}, and~\cite{AMM22} considered  non-adaptive settings for both counting and sampling. This model has since been extended to as well as to more general combinatorial frameworks~\cite{BBGM19,BGMP19,BBGM21}.

Furthermore, \cite{DL21,DLM22,DLM24} investigated the IS model alongside the Colorful Independence Oracle -- a generalization of the Bipartite Independent Set (BIS) oracle -- and developed efficient algorithms for approximate edge counting in hypergraphs. Their approach establishes fine-grained reductions from approximate counting to decision problems with minimal overhead. In particular, this work provides a framework for converting decision oracles into approximate counting and sampling algorithms for $k$-uniform hypergraphs and small witnesses, such as $k$-cliques. Closely related is the work of Beretta, Chakrabarty, and Seshadhri~\cite{BCS26}, who investigate edge counting in a model similar to our hybrid setting but with a more powerful ``global'' oracle.

\section{Conceptual overview}
In this work, we present two algorithms for edge sampling: a hybrid algorithm, which leverages both local (degree and neighbor) and independent-set (IS) queries, and a pure IS algorithm, which relies exclusively on independent-set queries.

We begin by presenting the hybrid algorithm. While it can be viewed as a synthesis of the local sampling techniques from \cite{er18} and our IS-based approach, the pure IS algorithm is arguably more technically demanding. Although conceptually simpler, the IS model requires us to statistically estimate vertex degree categories rather than retrieving them deterministically through local queries

The core strategy, following the framework in \cite{er18}, involves partitioning vertices into three degree categories - Low (L), Medium (M), and High (H) - and sampling edges within each category independently. We classify edges based on the degree profiles of their endpoints.

To facilitate this, we consider three ``natural'' edge sampling primitives, each of which may fail to return an edge in a single execution:
\begin{itemize}
\item \textbf{Sparse sampler:} A sparse subset of vertices is drawn; if the induced subgraph contains exactly one edge, that edge is returned as the sample.
\item \textbf{Path sampler:} A vertex is selected as a starting point for a short random walk, with the final edge traversed serving as the sample.
\item \textbf{Link sampler:} Two vertices are drawn independently, and an IS query is used to determine if they are adjacent.
\end{itemize}

In both the path and link samplers, we have an additional degree of freedom: the initial vertex distribution. We leverage the global reach of independent-set queries to prioritize MH-vertices.

Each edge category is associated with a specific sampling primitive for which the primitive is optimally suited.

In the general case -- specifically when the graph is not excessively sparse -- we employ the following strategy.

\begin{itemize}
\item \textbf{For L--L edges:} We utilize the sparse sampler, as the low density of these edges allows a sparse vertex subset to isolate a single edge with sufficient probability.
\item \textbf{For L--M edges:} We employ a one-step path sampler. To ensure efficiency, the starting vertex is sampled from a distribution that  prioritizes MH-vertices, thereby increasing the likelihood that the single-edge walk terminates at an L-vertex.
\item \textbf{For L--H edges:} We use a two-step path sampler starting from a uniform vertex distribution, following the approach in \cite{er18}. This length-2 walk is necessary to ``reach'' the relatively rare L--H edges from a random starting point.
\item \textbf{For MH--MH edges:} We use the link sampler. By biasing the vertex selection toward MH-vertices, we significantly increase the probability that two independently drawn vertices are adjacent.
\end{itemize}

If the graph is sparse, $m = O(n^{2/3})$, the optimal sampler for L--H edges is no longer the two-step walk from \cite{er18}, but rather the targeted one-step sampler used for L--M edges. The threshold $m \approx n^{2/3}$ is a natural transition point, as it represents the ``break-even'' equilibrium between $\sqrt{m}$ and $\sqrt{n/\sqrt{m}}$, the two terms that govern our query complexity.

\paragraph{Factors}
We introduce four \emph{factors} that characterize the bias of the natural samplers described above.

The first factor is the \emph{loneliness} factor, arising in the analysis of the sparse sampler. If we draw a set of density $p$, then we expect to have $p^2 m$ edges. If $p^2 m$ is smaller than one but still $\Omega(1)$, then the induced subgraph has exactly one edge with constant probability. While the probability of a specific edge $e$ to be in the induced subgraph is clearly bounded by $p^2$ (which is the probability that both $e$'s vertices belong to the vertex set), the probability that it is the \emph{only} edge is lower\footnote{Unless $e$ is the only edge in the graph.}. We denote this reduced probability by $\mathcal L_e \cdot p^2$, where $0 \le \mathcal L_e \le 1$ is the \emph{loneliness factor} of $e$. We can derive the following definition for $\mathcal L_e$ (depending on $p$), assuming that $e=(u,v)$, $G=(V,E)$ and $\bS \sim \Bin(V,p)$:
\[  \mathcal L_e
    = \Pr\left[\text{$G_{\bS\setminus\{u\}}$ has no edges} \wedge \text{$G_{\bS\setminus\{v\}}$ has no edges} \;\cond\; u,v \in \bS \right]
\]

The second factor is the \emph{starness} factor, arising in the attempt to prioritize drawing of MH-vertices. Again, assume that we draw a set of density $p$ where $p^2 m$ is smaller than one but still $\Omega(1)$. We consider the event that there is at least one edge in the induced subgraph, and all these edges are adjacent to a common vertex (note that in the case of a single edge, there are \emph{two} such ``common'' vertices). This event contains the high-probability event of having exactly one edge.

Consider the following process: first, we draw a sparse vertex set $\bS$ as before, and then we uniformly choose a vertex adjacent to all edges (if such a vertex exists and there is at least one edge). If the induced subgraph contains at least two edges, then such a vertex (if exists) is distinct, and therefore we choose it with probability $1$. We refer to such a vertex as a \emph{full-star} vertex. If the induced subgraph contains a single edge, then each of its endpoint vertices is chosen with probability $1/2$. These vertices are referred to as \emph{semi-star} vertices. The probability of a vertex $u$ to be chosen according to this process is bounded by $p$ (which is the probability that $u$ belongs to the vertex set), and can be denoted by $\mathcal S_u \cdot p$, where $0 \le \mathcal S_u \le 1$ is the \emph{starness factor} of $u$. We can derive the following definition for $\mathcal S_u$ (depending on $p$), assuming that $G=(V,E)$ and $\bS \sim \Bin(V,p)$:
\[  \mathcal S_u
    = \frac{1}{2} \Pr\left[\text{$G_{\bS\setminus\{u\}}$ has no edges} \wedge \abs{\bS \cap \mathrm{N}_u} = 1\right] + \Pr\left[\text{$G_{\bS\setminus\{u\}}$ has no edges} \wedge \abs{\bS \cap \mathrm{N}_u} \ge 2\right].
\]

We have additional two ad-hoc factors: the ``tininess factor'' of a vertex, which is only used in the analysis of the hybrid algorithm, is the fraction of tiny-degree neighbors, and the ``neighborhood factor'' of a vertex, which is only used in the analysis of the IS algorithm, is the probability to draw a neighbor under specific constraints.

\paragraph{Canceling the factors}
To achieve a uniform edge sample, we must normalize the output probabilities by counteracting the biases introduced by our low-level sampling primitives. Specifically, we must ``cancel out'' the loneliness, starness, and ad-hoc factors that skew the sampling distribution of certain edges. To this end, we introduce \proc{estimate-indicator-inv}, a procedure that enables ``division by $\E[\bX]$'' given sampling access to a Bernoulli random variable $\bX$. While seemingly simple in concept, this allows us to re-weight our samples to compensate for the factors

\paragraph{Combining the category-specific samplers} In the hybrid algorithm, each category-specific sampler is close to uniform when conditioned on sampling an edge of its prioritized category, and cannot sample edges of the wrong category (thanks to the degree oracle that allows to determine the exact category of a given edge). More precisely, for every category $A$, we compute a mass $\lambda_A$ so that every edge in the prioritized category is sampled with probability in the range $e^{\pm \eps} \cdot \lambda_A$. To uniformly sample an edge from the graph, we compute a mass $\lambda$, and randomly choose a category (or an immediate failure) so that each category $A$ is chosen with probability $\lambda/\lambda_A$. Clearly, $\lambda$ must be sufficiently small so that $e^{\eps} \lambda \sum_{A} (1/\lambda_A) \le 1$ (to match the worst-case sampling bound). This way, each edge is sampled with probability $e^{\pm \eps} \lambda$, which means that the resulting sample (when conditioned on success) is $O(\eps)$-uniform.

\paragraph{The IS algorithm} Conceptually, the IS algorithm can be viewed as an adaptation of the hybrid framework where the medium-degree category is subsumed into the high-degree category, and local oracle calls are simulated using independent-set queries.

From an analytical perspective, the primary challenge lies in the uncertainty of vertex categorization. Without access to a degree oracle, the algorithm cannot deterministically assign a vertex to a category. It can only perform statistical tests to estimate its degree. This necessitates a careful analysis of the algorithm's behavior for ``boundary'' vertices -- those whose degrees lie near the threshold between the Low (L) and High (H) categories. At the top level, the algorithm partitions edges into two primary classes:

\begin{itemize}
\item \textbf{For L--L edges:} We continue to utilize the sparse-vertex-set sampler.
\item \textbf{For LH--H edges:} We employ the path sampler (analogous to the L--M strategy in the hybrid algorithm), where the neighbor oracle is replaced by a simulation based on independent-set queries.
\end{itemize}

In both cases, once a candidate edge is obtained, the algorithm verifies the degree categories of its endpoints. If the estimated degrees do not match the target profiles, the sample is rejected, and the algorithm declares a failure.

\subsection{Additional details}
Both the hybrid algorithm and the IS algorithm star by obtaining an estimation $\tilde{\bm}$ of the number of edges.

\begin{lemma}[Theorem 1 in \cite{ahl26} - Rephrased]{external:estimate-edges-hybrid}
    There exists an algorithm whose input is a graph $G = (V,E)$ over $n$ vertices (known to the algorithm) and $m$ edges (unknown to the algorithm) that is accessible through the degree oracle, the neighbor oracle and the independent-set oracle, whose output $\tilde{\bm}$ is in the range $e^{\pm \eps} m$ with probability at least $2/3$. Moreover, the expected query complexity of this algorithm is $O\left(\frac{\log n}{\eps^{5/2}} \min\{\sqrt{m}, \sqrt{n/\sqrt{m}}\}\right)$.
\end{lemma}

\begin{lemma}[Theorem 1 in \cite{clw20} - Rephrased]{external:estimate-edges-IS}
    There exists an algorithm whose input is a graph $G = (V,E)$ over $n$ vertices (known to the algorithm) and $m$ edges (unknown to the algorithm) that is accessible through the independent-set oracle, whose output $\tilde{\bm}$ is in the range $e^{\pm \eps} m$ with probability at least $2/3$. Moreover, the query complexity of this algorithm is $O\left(\min\{\sqrt{m}, n/\sqrt{m}\} \poly(\log n, 1/\eps)\right)$.
\end{lemma}

We use a constant accuracy parameter, so that $\tilde{\bm} \in e^{\pm 1/10} m$ with high probability. We amplify the success probability from $2/3$ to $1-r$, for $r=\poly(\eps,1/n)$, by taking the median of $O(\log (1/r)) = O(\log(n/\eps))$ independent estimations. In total, the expected cost of obtaining $\tilde{\bm}$ is:
\[\begin{array}{ll}
     O\left(\min\{\sqrt{m}, \sqrt{n/\sqrt{m}}\} \log n \log (n/\eps)\right) & \text{for the hybrid algorithm,} \\
     O\left(\min\{\sqrt{m}, n/\sqrt{m}\} \poly(\log n) \log (n/\eps) \right) & \text{for the IS algorithm.}
\end{array}\]

Based on the estimation $\tilde{\bm}$ we define the threshold degrees:
\begin{itemize}
    \item $k_1 = 3\sqrt{\tilde{\bm}^{3/2} / n}$ -- an upper bound for tiny degrees in the non-sparse case of the hybrid algorithm.
    \item $k_2 = \sqrt{\tilde{\bm}}$ -- an upper bound for low degrees.
    \item $k_3 = \sqrt{n \sqrt{\tilde{\bm}}}$ -- an upper bound for medium degrees in the hybrid algorithm.
\end{itemize}
More precisely, the degree ranges in the hybrid algorithm are $[0, k_2]$ for low, $(k_2, k_3]$ for medium and $(k_3,n-1]$ for high. In the IS algorithm, the high range is $(k_2,n-1]$.

Note that, up to poly-logarithmic factors, $k_2$ corresponds to the threshold degree in \cite{er18} (sampling, local queries) and \cite{clw20} (counting, IS queries) and $k_3$ corresponds to the threshold degree in \cite{ahl26} (counting, local+IS queries).

\paragraph{Estimating an indicator inverse}
\defproc{estimate-indicator-inv}{Estimate-Indicator-Inverse}
The following lemma formalizes the mechanism of ``division by $\E[\bX]$'' when $\bX$ is a Bernoulli variable accessible through samples.

\begin{lemma}{estimate-indicator-inv}
    Consider a black-box procedure $\mathcal A$ that results in an indicator. Let $p_\mathcal A$ be the expected value of this indicator. There exists a procedure $\proc*{estimate-indicator-inv}(\mathcal A, \eps, \rho)$ that makes $O\left(\frac{1}{p_\mathcal A} \right)$ calls in expectation to $\mathcal A$ and:
    \begin{itemize}
        \item The outcome is a Bernoulli variable.
        \item If $p_\mathcal A \ge \rho$, then the expected output is in the range $e^{\pm \eps} \rho / (p_\mathcal A \cdot \ln (\paperZcoefZestZindZinv/\eps))$.
    \end{itemize}
   Moreover, the worst-case number of calls is $O\left(\log \eps^{-1} / \rho\right)$ (even if $p_\mathcal A < \rho$).
\end{lemma}
We prove Lemma \reflemma{estimate-indicator-inv} in Section \ref{sec:common}.

\paragraph{Finding an edge in a non-independent set}
\defproc{extract-edge-IS}{Extract-Edge-IS}

Some of our procedures require the ability to obtain an arbitrary edge from a given non-independent vertex set.

\begin{lemma}{extract-edge-IS}
    There exists a deterministic procedure $\proc*{extract-edge-IS}(G;S)$ whose input is a non-empty vertex set $S$ of a graph $G$ that makes $O(1 + \log \abs{S})$ independent-set queries (worst-case) and returns an arbitrary edge $e$ between two $S$-vertices, if such an edge exists, or \reject otherwise.
\end{lemma}
Note that this procedure also appears in \cite{beame20, clw20}. For the sake of self-containment, we provide a version of the algorithm in Section \ref{sec:common}.

\section{Preliminaries}

\subsection{Notation scheme}

\paragraph{Graph notations} The input graph is usually denoted by $G = (V, E)$. It has $n = \abs{V}$ vertices and $m = \abs{E}$ edges. At some point we define sets categorizing degree vertices: $L$ (low), $M$ (medium) and $H$ (high). We use $E_\mathrm{A,B}$ to denote the set of edges between category $A$ and category $B$. For a vertex set $S$, we use $E_S$ to denote the set of edges in the subgraph induced on the vertices of $S$ and $G_S$ to denote this induced subgraph.

\paragraph{Set of neighbors} We use $\mathrm{N}_u = \{ v \in V : \text{the edge $uv$ exists} \}$ to denote the set of neighbor vertices of $u$.

\paragraph{Random subsets} When drawing a random subset $\bS \subseteq U$ with \emph{density} $p$, every element in $U$ belongs to $\bS$ with probability $p$, independently. This drawing is also denoted by $\Bin(U,p)$.

\paragraph{Range notations} An expression of the form $\pm \alpha$ indicates the range between $-\alpha$ and $\alpha$. Composite ranges are defined elementwise (for example, the implicit range notation $\exp([\min,\max])$ indicates the range $[\exp(\min), \exp(\max)]$). Multiple occurrences of $\pm$ in the same expression are parsed individually.

\paragraph{Multiplicative errors} We denote bounded multiplicative errors by an $e^{\pm \eps}$-factor. Note that this is slightly different from the usual notation of $(1 \pm \eps)$-factor. Observe that $e^{\pm \eps/2} \subseteq 1 \pm \eps \subseteq e^{\pm 2\eps}$ for a sufficiently small $\eps$, and therefore, both notations are essentially the same up to a constant factor.

\paragraph{Pseudocode conventions} In a few algorithms we use the following phrases:
\begin{itemize}
    \item Filter by [calling a procedure whose output is a Bernoulli variable].
    \item Proceed with probability $p$.
\end{itemize}
The first form requires calling the specified procedure and terminate if the result is $0$ (or \reject), and the second form terminates with probability $1-p$. In both cases, termination is done by returning a value indicating failure (usually \reject). We believe that these declarative phrases improve the readability over the pure-imperative if-then scheme.

\paragraph{Query model} Across the paper, every procedure describes the graph oracles it may access as a part of its name (IS, Local, Hybrid).

\begin{definition}[An $\eps$-uniform distribution]{eps-uniform-distribution}
    A distribution $\mu$ over a discrete domain $\Omega$ is \emph{$\eps$-uniform} if $\mu(x) \in e^{\pm \eps} / \abs{\Omega}$ for every $x \in \Omega$.
\end{definition}

\begin{definition}[Edge sampling algorithm]{edge-sampling-algorithm}
    An algorithm is a \emph{$(\lambda,\eps)$-edge-sampling algorithm} if, for every input graph $G$ that has at least one edge:
    \begin{itemize}
        \item With probability at least $\lambda$, the output is not a declaration of failure.
        \item When conditioned on success, the output of the algorithm is $\eps$-uniform over the edge set of $G$.
    \end{itemize}
\end{definition}

\subsection{A few technical lemmas and tools}
The following technical lemmas are well known and easily verifiable (either explicitly or using automated tools such as Wolfram Alpha).

\begin{lemma}[Technical lemma]{ubnd-1+x}
    For every $x$, $1 + x \le e^x$.
\end{lemma}

\begin{lemma}[Technical lemma]{ubnd-exp-minus-x}
    For every $x \ge 0$, $e^{-x} \le 1 - x + \frac{1}{2}x^2$.
\end{lemma}

\begin{lemma}[Technical lemma]{lbnd-1-plus-x-power-y}
    For every $x \ge -1$ and $y \ge 1$, $(1 + x)^y \ge (1 - x^2 y)e^{xy}$.
\end{lemma}

\begin{lemma}[Chernoff bounds]{chernoff-bounds}
    Let $\bX$ be a random variable distributing Binomially or according to Poisson distribution.
    \[ \begin{array}{rlclcl}
        [0 < \delta \le 1] & \Pr\left[\bX \le (1 - \delta)\E[\bX]\right] &\le& e^{-\frac{1}{2}\delta^2 \E[\bX]}\\{}
        [\delta > 0] & \Pr\left[\bX \ge (1 + \delta)\E[\bX]\right] &\le& e^{-\frac{\delta^2}{2+\delta} \E[\bX]}
        \le \begin{array}[t]\{{ll}.
            e^{-\frac{1}{3}\delta^2} & \text{if $\delta \le 1$} \\
            e^{-\frac{1}{3}\delta} & \text{if $\delta \ge 1$}
        \end{array}
    \end{array} \]
\end{lemma}

\section{Common samplers}

In this section we provide two samplers that are used in both the hybrid algorithm and the IS algorithm. The other samplers mentioned in the overview are described within each algorithm's detail section.

\subsection{Lone-edge sampler}
\defproc{sample-lone-edge-IS}{Sample-Lone-Edge-IS}

The lone-edge sampler, which corresponds to the ``sparse-vertex-set sampler'' in the conceptual overview, draws a vertex set $\bS$ of density $p = 1/10\sqrt{\tilde{m}}$. If the induced subgraph $G_{\bS}$ has exactly a single edge, then the algorithm returns it. Otherwise it returns \reject.

Algorithm \refalg{sample-lone-edge-IS} (\proc*{sample-lone-edge-IS}) provides the pseudocode for the sparse-set sampler. In the following we define $\mathcal L_e(\tilde{m})$, the \emph{loneliness factor} of $e$ (as a function of $\tilde{m}$), and show that the probability to sample $e$ using this algorithm is exactly $\frac{1}{100\tilde{m}} \mathcal L_e(\tilde{m})$.

\begin{proc-algo}{sample-lone-edge-IS}{G,n; \tilde{m}}
    \algoutput[noperiod]{Every edge $uv \in E$ is returned with probability exactly $\frac{1}{100\tilde{m}} \mathcal L_{u,v}(\tilde{m})$}
    \algcomplexity{$O(\log n)$ worst-case}
    \begin{code}
        \algitem Let $p \gets 1 / 10 \sqrt{\tilde{m}}$.
        \algitem Let $\bS \subseteq V$ be a set that every vertex belongs to with probability $p$ iid.
        \algpushcomment{(No edges)}
        \begin{If}{$\oracleIS(\bS)$ accepts}
            \algitem Return \reject.
        \end{If}
        \algitem Let $uv \gets \proc{extract-edge-IS}(G;\bS)$.
        \algpushcomment{(More than one edge in $\bS$)}
        \begin{If}{$\proc{test-loneliness-IS}(\bS,u,v)$}
            \algitem Return $uv$.
        \end{If}
        \algitem Return \reject.
    \end{code}
\end{proc-algo}

\defproc{test-loneliness-IS}{Test-Loneliness-IS}
We say that an unordered pair $\{u,v\}$ (regardless of whether it forms an edge) is \emph{lonely} in a vertex set $S$, if there are no edges in $S \cup \{u,v\}$, possibly except the edge $uv$ itself. Algorithm \refalg{test-loneliness-IS} (\proc*{test-loneliness-IS}) determines whether the pair $\{u,v\}$ is lonely in a vertex set $S$ directly by definition.

\begin{proc-algo}{test-loneliness-IS}{S; u,v}
    \algcomplexity{$O(1)$ worst-case}
    \begin{code}
        \begin{If}{$\oracleIS(S \cup \{v\} \setminus \{u\})$ rejects}
            \algitem Return \reject. \algcomment Has a non-$u$ edge.
        \end{If}
        \begin{If}{$\oracleIS(S \cup \{u\} \setminus \{v\})$ rejects}
            \algitem Return \reject. \algcomment Has a non-$v$ edge.
        \end{If}
        \algitem Return \accept.
    \end{code}
\end{proc-algo}

\begin{observation}{test-loneliness-IS}
    Procedure \proc*{test-loneliness-IS} makes $O(1)$ IS queries at worst-case.
\end{observation}

\defproc{loneliness-event-IS}{Lonliness-Event-IS}
For an edge $e=uv$, we define the \emph{loneliness} factor of $e$ as the ratio between the probability that $E_{\bS} = \{e\}$ and the probability that $e \in E_{\bS}$, where $\bS \sim \Bin(V,p)$ and $p=1/10\sqrt{\tilde{m}}$. Alternatively, this is the probability that the pair $\{u,v\}$ is lonely in the drawn set $\bS$.

To draw an indicator whose expected value is $\mathcal L_{u,v}(\tilde{m})$, we draw a set $\bS \subseteq V$ of density $p = 1/10\sqrt{\tilde{m}}$ and test whether or not $\{u,v\}$ is lonely in $\bS$. Algorithm \refalg{loneliness-event-IS} (\proc*{loneliness-event-IS}) provides the pseudocode for this logic.

\begin{proc-algo}{loneliness-event-IS}{G,n; \tilde{m}, u, v}
    \algoutput{Binary answer. Accept probability $\mathcal L_{u,v}(\tilde{m})$}
    \algcomplexity{$O(1)$ worst-case}
    \begin{code}
        \algitem Let $p \gets 1 / 10 \sqrt{\tilde{m}}$.
        \algitem Let $\bS \subseteq V$ be a set that every vertex belongs to with probability $p$ iid.
        \algitem Accept if and only if $\proc{test-loneliness-IS}(\bS;u,v)$ accepts.
    \end{code}
\end{proc-algo}

\begin{observation}{loneliness-event-IS}
    A call to $\proc*{loneliness-event-IS}(G,n;\tilde{m},u,v)$ accepts the input with probability exactly $\mathcal L_{u,v}(\tilde{m})$ at the cost of $O(1)$ IS-queries worst-case.
\end{observation}

At this point we have all the required notations to state the exact behavior of the lone-edge sampler.

\begin{lemma}{sample-lone-edge-IS}
    For every edge $u_0 v_0 \in E$, the probability that \proc*{sample-lone-edge-IS} returns $u_0 v_0$ is exactly $\frac{1}{100\tilde{m}} \mathcal L_{u_0,v_0}(\tilde{m})$. Moreover, the query complexity is $O(\log n)$ worst-case.
\end{lemma}
\begin{proof}
    For complexity, observe that there is one explicit independent-set query, one call to \proc{extract-edge-IS} at the cost of $O(\log n)$ IS-queries (Lemma \reflemma{extract-edge-IS}) and one call to \proc{test-loneliness-IS} at the cost of $O(1)$ IS-queries (Observation \refobs{test-loneliness-IS}).

    The probability to return the edge $u_0 v_0$ is exactly:
    \begin{eqnarray*}
        \Pr\left[\text{return $u_0 v_0$}\right]
        = \Pr\left[u_0 \in \bS\right] \cdot \Pr\left[v_0 \in \bS\right] \cdot \Pr\left[\text{the pair $u_0, v_0$ is lonely in $\bS$}\right]
        &=& p \cdot p \cdot \mathcal L_{u_0,v_0}(\tilde{m}) \\
        &=& \frac{1}{100\tilde{m}} \mathcal L_{u_0,v_0}(\tilde{m})
    \end{eqnarray*}
\end{proof}

\subsection{Star-vertex sampler}
\defproc{sample-star-vertex-IS}{Sample-Star-Vertex-IS}

The star-vertex sampler, which corresponds to the vertex sampler in the conceptual overview that favors medium-degree and high-degree vertices, starts with drawing a set $\bS$ of density $p = 1/10\sqrt{\tilde{m}}$, and looks for the following cases: (1) the induced subgraph $G_{\bS}$ has a single edge, in which case we uniformly choose one of its endpoint vertices as the sample, and (2) the induced subgraph $G_{\bS}$ has two or more edges, all adjacent to a single vertex (a ``star''), which we choose as our sample. In every other case, the sampler return \reject.

Algorithm \refalg{sample-star-vertex-IS} (\proc*{sample-star-vertex-IS}) provides the pseudocode for the star-vertex sampler. In the following we define $\mathcal S_u(\tilde{m})$, the \emph{starness factor} of $u$ (as a function of $\tilde{m}$), and show that the probability to sample $u$ using this algorithm is exactly $\frac{1}{10\sqrt{\tilde{m}}} \mathcal S_u(\tilde{m})$.

\begin{proc-algo}{sample-star-vertex-IS}{G,n; \tilde{m}}
    \algoutput[noperiod]{Every vertex $u \in V$ is returned with probability exactly $\frac{1}{10\sqrt{\tilde{m}}} \mathcal S_u(\tilde{m})$}
    \algcomplexity{$O(\log n)$}
    \begin{code}
        \algitem Let $p \gets 1 / 10 \sqrt{\tilde{m}}$.
        \algitem Let $\bS \subseteq V$ be a set that every vertex belongs to with probability $p$ iid.
        \algpushcomment{(No edges)}
        \begin{If}{$\oracleIS(\bS)$ accepts}
            \algitem Return \reject.
        \end{If}
        \algitem Let $\bz\bw \gets \proc{extract-edge-IS}(G,\bS)$.
        \algitem Exchange $\bw$ and $\bz$ with probability $1/2$.
        \algpushcomment{($\bw$ connected to all edges)}
        \begin{Elif}{$\oracleIS(\bS \setminus \{\bw\})$ accepts}
            \algitem Return $\bw$.
        \end{Elif}
        \algpushcomment{($\bz$ connected to all edges, which $w$ does not)}
        \begin{Elif}{$\oracleIS(\bS \setminus \{\bz\})$ accepts}
            \algitem Return $z$.
        \end{Elif}
        \algitem Return \reject.
    \end{code}
\end{proc-algo}

\defproc{test-starness-IS}{Test-Starness-IS}
Algorithm \refalg{test-starness-IS} (\proc*{test-starness-IS}) gets an input $(S;u)$ and accepts with probability equals to the probability to sample $u$ (by \proc*{sample-star-vertex-IS}) when drawing the set $\bS \cup \{u\}$.

\begin{proc-algo}{test-starness-IS}{S; u}
    \algcomplexity{$O(\log n)$ worst-case}
    \begin{code}
        \begin{If}{$\oracleIS(S)$ rejects}
            \algitem Return \reject. \algcomment{(non-$u$ edges exist)}
        \end{If}
        \begin{If}{$\oracleIS(S \cup \{u\})$ accepts}
            \algitem Return \reject. \algcomment{(no edges at all)}
        \end{If}
        \algitem Let $wz \gets \proc{extract-edge-IS}(G,S \cup \{u\})$.
        \begin{If}{$\proc{test-loneliness-IS}(S,w,z)$}
            \algitem Return \accept with probability $\frac{1}{2}$, otherwise \reject. \algcomment{(a single edge)}
        \end{If}
        \algitem Return \accept. \algcomment{(two or more edges)}
    \end{code}
\end{proc-algo}

\begin{observation}{test-starness-IS}
    Procedure \proc*{test-starness-IS} makes $O(\log n)$ IS queries at worst-case.
\end{observation}

\defproc{starness-event-IS}{Starness-Event-IS}
To draw an indicator whose expected value is $\mathcal S_u(\sqrt{\tilde{m}})$, we draw a set $\bS \subseteq V$ of density $p = 1/10\sqrt{\tilde{m}}$ and then accept according to the starness-test of $u$ with respect to $\bS$. Algorithm \refalg{starness-event-IS} (\proc*{starness-event-IS}) provides the pseudocode for this logic.

\begin{proc-algo}{starness-event-IS}{G,n; \tilde{m}, u}
    \algoutput{Binary answer. Accept probability $\mathcal S_u(\tilde{m})$}
    \algcomplexity{$O(\log n)$ worst-case}
    \begin{code}
        \algitem Let $p \gets 1 / 10 \sqrt{\tilde{m}}$.
        \algitem Let $\bS \subseteq V$ be a set that every vertex belongs to with probability $p$ iid.
        \algitem Accept if and only if $\proc{test-starness-IS}(\bS; u)$ accepts.
    \end{code}
\end{proc-algo}

\begin{observation}{starness-event-IS}
    A call to $\proc*{starness-event-IS}(G,n;\tilde{m},u)$ accepts the input with probability $\mathcal S_u(\tilde{m})$ at the cost of $O(\log n)$ IS-queries worst-case.
\end{observation}

At this point we have all the required notations to state the exact behavior of the star-vertex sampler.

\begin{lemma}{sample-star-vertex-IS}
    Algorithm \refalg{sample-star-vertex-IS} (\proc*{sample-star-vertex-IS}) samples every vertex $u$ with probability $\frac{1}{10\sqrt{\tilde{m}}} \mathcal S_u(\tilde{m})$, and otherwise rejects. Also, the query complexity is $O(\log n)$ worst-case.
\end{lemma}

\begin{proof}
    For complexity, observe that there are three independent-set queries, in additional a single call to \proc{extract-edge-IS}, at the cost of $O(\log n)$ independent set queries at worst-case (Lemma \reflemma{extract-edge-IS}).

    For a vertex $u$, the probability to return it is:
    \begin{eqnarray*}
        \Pr\left[\text{return $u$}\right]
        &=& \Pr\left[u \in \bS\right] \left(\begin{array}{l}
            \frac{1}{2}\Pr\left[\text{single edge, adjacent to $u$} \cond u \in \bS\right] + \\
            \phantom{\frac{1}{2}} \Pr\left[\text{more than one edge, all adjacent to $u$} \cond u \in \bS\right]
            \end{array}\right)\\
        &=& \Pr\left[u \in \bS\right] \Pr\left[\proc{starness-event-IS}(G,n;\tilde{m},u)\right] \\
        &=& p \cdot \mathcal S_u(\tilde{m})\\
        &=& \frac{1}{10\sqrt{\tilde{m}}} \mathcal S_u(\tilde{m})
    \end{eqnarray*}
\end{proof}

\subsection{Factor lower bounds}
We show settings in which the factors are lower-bounded by a constant.

The lone-edge sampler favors low-degree vertices, as stated in the following lemma. To be applicable to the IS algorithm, the lemma statement also includes degrees slightly higher than the threshold between low and medium degrees.

\begin{lemma}{lbnd-loneliness}
    Recall that $k_2 = \sqrt{\tilde{m}}$, and assume that $\tilde{m} \ge e^{-1/10} m$. If $\deg(u) \le 2 k_2$ and $\deg(v) \le 2 k_2$, then $\mathcal L_{u,v}(\tilde{m}) \ge \frac{1}{2}$.
\end{lemma}
\begin{proof}
    The loneliness event is the negation of the union of the following events:
    \begin{itemize}
        \item $\bS$ has a $u$-neighbor which is not $v$.
        \item $\bS$ has a $v$-neighbor which is not $u$.
        \item $\bS$ has an edge adjacent to neither $u$ nor $v$.
    \end{itemize}

    The probability to have a $u$-neighbor (which is not $v$) is bounded by $p \cdot (\deg(u) - 1) \le p \deg(u) \le (1/10\sqrt{\tilde{m}})\cdot 2\sqrt{\tilde{m}} = 1/5$. Analogously, this is also a bound for the probability to have a $v$-neighbor (which is not $u$).

    For every edge that is adjacent to neither $u$ nor $v$, the probability that both its vertices belong to $\bS$ is $p^2 = 1 / 100\tilde{m}$. There are at most $m$ such edges, and therefore, the expected number of these edges in $\bS$ is bounded by $\frac{1}{100}(m / \tilde{m}) \le e^{1/10} / 100$.

    By linearity of expectation, the expected number of ``bad'' edges in $\bS$ is bounded by $\frac{1}{5} + \frac{1}{5} + \frac{e^{1/10}}{100} < \frac{1}{2}$. By Markov's inequality, the probability to have ``bad'' edges is at least $1/2$.
\end{proof}

The star-vertex sampler favors medium-degree and high-degree vertices, as stated in the following lemma.

\begin{lemma}{lbnd-starness}
    Recall that $k_2 = \sqrt{\tilde{m}}$, and assume that $\tilde{m} \ge \max\{4, e^{-1/10} m\}$. If $\deg(u) \ge k_2$ then $\mathcal S_u(\tilde{m}) \ge 1/30$.
\end{lemma}
\begin{proof}
    The starness event is the negation of the union of the following three events:
    \begin{itemize}
        \item $\bS$ has an edge non-adjacent to $u$.
        \item $\bS$ has no $u$-neighbors.
        \item $\bS$ has a single $u$-neighbor, and an additional uniform bit is $1$.
    \end{itemize}
    (Note that the ``uniform bit'' represents the symmetry described in the overview in the case where there is only one edge)

    The number of edges not-adjacent to $u$ is at most $m$, and each of them belongs to $\bS$ with probability $p^2 = 1/100\tilde{m}$. The expected number of these ``bad'' edges is bounded by $p^2 m = (1/100)(m / \tilde{m}) \le e^{1/10} / 100$. By Markov's inequality, this is also an upper bound for the probability to have any of these edges in $\bS$.

    For the other two events, we use an explicit bound:
    \begin{eqnarray*}
        \Pr\left[\Bin(\deg(u), p) = 0\right] + \frac{1}{2}\Pr\left[\Bin(\deg(u), p) = 1\right]
        = (1 - p)^{\deg(u)} + \frac{1}{2}p \deg(u) (1 - p)^{\deg(u) - 1}
    \end{eqnarray*}

    Since the function $x \to (1-p)^x + \frac{1}{2}px(1-p)^{x-1}$ is decreasing monotone, we can use $\deg(u) \ge \sqrt{\tilde{m}}$ to obtain:
    \begin{eqnarray*}
        [\cdots] \le (1 - p)^{\sqrt{\tilde{m}}} + \frac{1}{2}p \sqrt{\tilde{m}} (1 - p)^{\sqrt{\tilde{m}} - 1}
        &=& (1 - p)^{\sqrt{\tilde{m}}} \left(1 + \frac{1}{2(1 - p)}p \sqrt{\tilde{m}}\right) \\
        &\le& e^{-1/10} \left(1 + \frac{1}{20(1 - p)}\right) \\
        \text{[Since $\tilde{m} \ge 4 \to p \le 1/20$]} &\le& e^{-1/10} \left(1 + \frac{1}{20(1 - 1/20)}\right)
        \le 0.953
    \end{eqnarray*}

    Therefore, $\mathcal S_u(\tilde{m}) \ge 1 - e^{1/10}/100 - 0.953 \ge 1 / 30$.
\end{proof}

\section{Elementary procedures}
\label{sec:common}

In this section we describe elementary procedures which we use in our algorithms. This section is mostly provided for technical completeness (since the front-end interface is already stated in the overview), and has a few to do with our main contribution.

\subsection{Estimating the inverse of an indicator}
Algorithm \refalg{estimate-indicator-inv} (\proc*{estimate-indicator-inv}) gets a black-box $\mathcal A$ for drawing an indicator whose expected value is denoted by $p_\mathcal A$, an accuracy parameter $\eps$ and a saturation parameter $\rho$. Its output is $1$ with probability $\min\{X \cdot \rho/\ln(\paperZcoefZestZindZinv/\eps), 1\}$ (and $0$ otherwise), where $X \sim \Geo(p_\mathcal A)$ is obtained by repeatedly calling the black-box procedure $\mathcal A$ until success.

\begin{proc-algo}{estimate-indicator-inv}{\mathcal A, \eps, \rho}
    \begin{code}
        \algitem Let $C \gets \ln (\paperZcoefZestZindZinv/\eps)$.
        \algitem Let $N_\mathrm{max} \gets \floor{C/\rho}$.
        \algitem Set $\bN \gets 0$.
        \begin{While}{$N < N_\mathrm{max}$}
            \algitem Set $N \gets N + 1$.
            \algitem Call $\mathcal A$, giving $\mathit{ans}$.
            \begin{If}{$\mathit{ans}$ is \accept}
                \algitem Break loop.
            \end{If}
        \end{While}
        \algitem Return a bit distributing like $\Ber(\bN \cdot (\rho / C))$.
    \end{code}
\end{proc-algo}

\repprovelemma{estimate-indicator-inv}
\begin{proof}
    Let $X \sim \Geo(p_\mathcal A)$. Clearly, $N$ distributes as $\min\{X, N_\mathrm{max}\}$. Since $N_\mathrm{max} \le C/\rho$, the output is always bounded between $0$ and $1$.
    
    The expected value of $N$ is:
    \begin{eqnarray*}
        \E[N]
        = \E\left[\min\{X, N_\mathrm{max}\}\right]
        &=& \E[X] - \Pr\left[X \ge N_\mathrm{max} + 1\right]\E\left[X - N_\mathrm{max} \cond X \ge N_\mathrm{max} + 1 \right] \\
        (*) &=& \E[X] - \Pr\left[X \ge N_\mathrm{max} + 1 \right]\E\left[X\right] \\
        &=& \left(1 - \Pr\left[X \ge N_\mathrm{max} + 1\right]\right) \E[X]
    \end{eqnarray*}
    Where the $(*)$-transition is correct since geometric variable is memoryless.

    If $p_\mathcal A \ge \rho$, then:
    \begin{eqnarray*}
        0
        \le \Pr\left[X \ge N_\mathrm{max} + 1\right]
        &=& (1 - p_\mathcal A)^{N_\mathrm{max}} \\
        &\le& e^{-p_\mathcal A \floor{C/\rho}} \\
        \text{[Since $p_{\mathcal A} \ge \rho$]} &\le& e^{-p_{\mathcal A} (C/p_{\mathcal A} - 1)}
        = e^{-C + p_\mathcal A}
        \le e^{-\ln (\paperZcoefZestZindZinv/\eps) + 1}
        = \frac{e}{\paperZcoefZestZindZinv}\eps
        \le 1 - e^{-\eps}
    \end{eqnarray*}

    Therefore, if $p_\mathcal A \ge \rho$, then:
    \begin{eqnarray*}
        \E[\text{output}]
        = e^{\pm \eps} \E[X] \cdot \frac{\rho}{C}
        = e^{\pm \eps} \cdot \frac{1}{p_\mathcal A} \cdot \frac{\rho}{C}
        = e^{\pm \eps} \cdot \frac{\rho}{p_\mathcal A \ln (\paperZcoefZestZindZinv/\eps)}
    \end{eqnarray*}

    The expected cost is bounded by $\E[X] = O(1/p_\mathcal A)$. The worst-case cost is $N_\mathrm{max} = O(C/\rho) = O(\log \eps^{-1} / \rho)$.
\end{proof}

\subsection{Extracting an edge}

Here we provide pseudocode and proof for the \proc*{extract-edge-IS} procedure. Note that this tool appears in \cite{beame20} (implicit) and in \cite{clw20} (explicit randomized version). For self containment, we provide it as Algorithm \refalg{extract-edge-IS} and then prove its correctness as an implementation for Lemma \reflemma{extract-edge-IS}.

Given a non-independent set $S$, the algorithm repeatedly chooses a proper non-independent subset of $S$ until reaching a non-independent set of exactly two vertices, which explicitly describes an edge.

\begin{proc-algo}{extract-edge-IS}{G;S}
    \alginput{A graph $G$, a vertex set $S$}
    \algoutput{An arbitrary edge between two $S$-vertices or \reject if $S$ is independent}
    \algcomplexity{$O(\log \abs{S})$}
    \begin{code}
        \begin{If}{$\abs{S} \le 1$}
            \algitem Return \reject.
        \end{If}
        \begin{If}{$\oracleIS(S)$ accepts}
            \algitem Return \reject.
        \end{If}
        \algitem Let $S_1 \cup S_2$ be an arbitrary partition of $S$ of sizes $\ceil{\abs{S}/2}$ and $\floor{\abs{S}/2}$.
        \begin{If}{$\oracleIS(S_1)$ rejects}
            \algitem Return $\proc*{extract-edge-IS}(G;S_1)$.
        \end{If}
        \begin{If}{$\oracleIS(S_2)$ rejects}
            \algitem Return $\proc*{extract-edge-IS}(G;S_2)$.
        \end{If}
        \begin{While}{$\abs{S_1} > 1$}
            \algitem Let $S_{11} \cup S_{12}$ be an arbitrary partition of $S_1$ of sizes $\ceil{\abs{S_1}/2}$ and $\floor{\abs{S_1}/2}$.
            \begin{If}{$\oracleIS(S_{11} \cup S_2)$ accepts}
                \algitem Set $S_1 \gets S_{12}$.
            \end{If}
            \begin{Else}
                \algitem Set $S_1 \gets S_{11}$.
            \end{Else}
        \end{While}
        \begin{While}{$\abs{S_2} > 1$}
            \algitem Let $S_{21} \cup S_{22}$ be an arbitrary partition of $S_2$ of sizes $\ceil{\abs{S_2}/2}$ and $\floor{\abs{S_2}/2}$.
            \begin{If}{$\oracleIS(S_1 \cup S_{22})$ accepts}
                \algitem Set $S_2 \gets S_{21}$.
            \end{If}
            \begin{Else}
                \algitem Set $S_2 \gets S_{22}$.
            \end{Else}
        \end{While}
        \algitem Let $u \in S_1$, $v \in S_2$.
        \algitem Return $uv$.
    \end{code}
\end{proc-algo}

We recall Lemma \reflemma{extract-edge-IS} and prove it.

\repprovelemma{extract-edge-IS}

\begin{proof}
    For correctness: if $\abs{S} \le 1$, then there is no edge to extract. For a partition $S_1 \cup S_2$, there are three cases (possibly overlapping): there is an edge in $G_{S_1}$, there is an edge in $G_{S_2}$, there is an edge crossing the cut $S_1 \cup S_2$. In the first two cases, the (tail-)recursive call is correct. If only the third case holds, then we deduce that both $S_1$ and $S_2$ are independent, but there exists an edge between them. In every additional step, we partition $S_1$ (or $S_2$) and focus on a subset that has an edge to $S_2$ (or $S_1$), until both $S_1$ and $S_2$ are singletons that have an edge between them.
    
    For complexity: observe that the sequence $n_{i+1} = \ceil{n_i / 2}$ converges to $1$ after $O(\log n)$ steps (for an integer $n_0 \ge 1$). Therefore, after $O(\log n)$ recursive steps, the complexity of the loop-part of the algorithm algorithm is $O(\log \abs{S_1} + \log \abs{S_2}) = O(\log S)$.
\end{proof}

\subsection{Brute-force sampler}
\defproc{sample-edge-bruteforce-IS}{Sample-Edge-Bruteforce-IS}
Some of our category-specific samplers assume that $\tilde{m}$ is sufficiently large. Therefore, we have to handle separately the case where there are only a few edges in the graph ($m = O(1)$).

We recall the following lemma:
\begin{lemma}[Exact count in \cite{beame20}, \textsf{Enumerate-Edges} in \cite{ahl26}]{enumerate-edges}
    There exists a deterministic algorithm that enumerates all edges in the given graph at the cost of $O(1 + m \log n)$ independent-set queries.
\end{lemma}

The brute-force algorithm enumerates all edges and then uniformly chooses one of them.
\begin{observation}{sample-edge-bruteforce-IS}
    There exists an algorithm $\proc*{sample-edge-bruteforce-IS}(G)$ that uniformly samples an edge in a given graph $G$ at the cost of $O(1 + m \log n)$ independent-set queries.
\end{observation}

\section{Using both IS and local queries}
\label{sec:IS-LOCAL}

In this section we provide the hybrid estimation algorithm. Conceptually, it can be seen as a mixture of the IS algorithm (in Section \ref{sec:IS-only}) and the local algorithm (see \cite{er18}).

For a small error $r = \min\{\eps^2 / 6, 1/n^{1/3} \log n \log^2 (1/\eps) \}$, we first estimate $\tilde{m} = e^{\pm 1/10} m$ with probability $1-r$. Since we can use both independent-set queries and local queries, the cost of this preprocess is $O(\min\{\sqrt{m}, \sqrt{n/\sqrt{m}}\} \cdot \log n \log (1/r))$ \cite{ahl26}. We use this small error to bound the contribution of the wrong-$\tilde{m}$ case to the sampling probability of each individual edge.

We recall the threshold degrees:
\begin{itemize}
    \item $k_1 = 3 \sqrt{\tilde{m} \sqrt{\tilde{m}} / n}$.
    \item $k_2 = \sqrt{\tilde{m}}$.
    \item $k_3 = \sqrt{n\sqrt{\tilde{m}}}$.
\end{itemize}
Note that $k_2 \le k_3$ (if we enforce the trivial bound $\tilde{m} \le n^2$), but the inequality $k_1 \le k_2$ does not always hold. We classify degrees as \emph{low} ($0 \le \deg(u) \le k_2$), \emph{medium} ($k_2 < \deg(u) \le k_3$) and \emph{high} ($\deg(u) > k_3$).

Using L, M, H to indicate the degree categories, our edge categories are:
\begin{itemize}
    \item L-L: between vertices of degree at most $k_2$ (low).
    \item L-M: between vertices of degree at most $k_2$ (low) and vertices of degree between $k_2$ and $k_3$ (medium).
    \item L-H: between vertices of degree at most $k_2$ (low) and vertices of degree greater than $k_2$ (high).
    \item MH-MH: between vertices of degree greater than $k_2$ (medium+high).
\end{itemize}
When $\tilde{m} \le n^{2/3}$, which corresponds to the anomaly case in \cite{ahl26}, we use an alternative logic that treats the categories L-M and L-H together as an ``L-MH'' category.

Note that, using the degree oracle, we can certainly determine the degree category of a given vertex.

\subsection{Sampling L-L edges}
\defproc{sample-LL-edge-hyb}{Sample-L-L-Edge-Hybrid}
Algorithm \refalg{sample-LL-edge-hyb} (\proc*{sample-LL-edge-hyb}) merely samples a lone edge, verifies that its vertices have low degree and normalizes the probability to return each edge by canceling the loneliness factor.

\begin{proc-algo}{sample-LL-edge-hyb}{G,n,\eps;\tilde{m}}
    \alginput{$\tilde{m} \ge e^{-1/10} m$}
    \algoutput{For every edge $e \in E_\mathrm{L,L}$, the probability to return $e$ is $\frac{e^{\pm \eps}}{200\tilde{m} \ln (\paperZcoefZestZindZinv/\eps)}$}
    \algcomplexity{$O(\log n)$ (expected)}
    \algcomplexity{$O(\log n + \log \eps^{-1})$ (worst-case)}
    \begin{code}
        \algitem Let $k_2 \gets \sqrt{\tilde{m}}$.
        \algitem Let $\bu\bv \gets \proc{sample-lone-edge-IS}(G,n;\tilde{m})$.
        \begin{If}{$\bu\bv$ is \reject}
            \algitem Return \reject.
        \end{If}
        \begin{If}{$\oracleDEG(\bu) > k_2$ or $\oracleDEG(\bv) > k_2$}
            \algitem Return \reject.
        \end{If}
        \algitem Filter by $\proc{estimate-indicator-inv}(\proc{loneliness-event-IS}(G,n;\tilde{m},\bu,\bv), \eps, 1/2)$.
        \algitem Return $\bu\bv$.
    \end{code}
\end{proc-algo}

\begin{lemma}{sample-LL-edge-hyb}
    If $\tilde{m} \ge e^{-1/10}$, then Algorithm \refalg{sample-LL-edge-hyb} (\proc*{sample-LL-edge-hyb}) samples every edge $e \in E_{L,L}$ with probability $e^{\pm \eps}/{200 \tilde{m} \ln (\paperZcoefZestZindZinv/\eps)}$, and otherwise rejects. Also, regardless of $\tilde{m}$, the query complexity is $O(\log n)$ in expectation and $O(\log n + \log \eps^{-1})$ in worst-case.
\end{lemma}
\begin{proof}
    For complexity, observe that there are at most two $\oracleDEG$-calls and one call to \proc{sample-lone-edge-IS}, which costs $O(\log n)$ $\oracleIS$-calls at worst-case (Lemma \reflemma{sample-lone-edge-IS}).

    If we reach the call to \proc{estimate-indicator-inv}, then $\deg(u) \le k_2$ and $\deg(v) \le k_2$. In this case, $\mathcal L_{u,v}(\tilde{m}) \ge \frac{1}{2}$ (Lemma \reflemma{lbnd-loneliness}), and therefore, this call costs $O(1)$ calls to \proc{loneliness-event-IS} in expectation and $O(\log \eps^{-1})$ times in worst-case (Lemma \reflemma{estimate-indicator-inv}). The cost for every such call is $O(1)$ (Observation \refobs{loneliness-event-IS}).
    
    Clearly, only L-L edges can be returned. By Lemma \reflemma{sample-lone-edge-IS}, every L-L edge $e = uv$ is returned with probability $\frac{1}{100\tilde{m}}\mathcal L_e(\tilde{m})$. The probability to pass the filter is in the range $e^{\pm \eps} \cdot \frac{1/2}{\mathcal L_e(\tilde{m}) \ln (\paperZcoefZestZindZinv/\eps)}$ (Lemma \reflemma{estimate-indicator-inv}). Combined, the probability to return any individual edge $e=\{u,v\}$ for which $\deg(u), \deg(v) \le k_2$ is in the range
    \[  \frac{1}{100\tilde{m}} \mathcal L_{u,v}(\tilde{m}) \cdot e^{\pm \eps} \frac{1/2}{\mathcal L_{u,v}(\tilde{m}) \ln (\paperZcoefZestZindZinv/\eps)}
        = \frac{e^{\pm \eps}}{200\tilde{m} \ln (\paperZcoefZestZindZinv/\eps)} \qedhere \]
\end{proof}

\subsection{Sampling medium-high vertices}
\defproc{sample-MH-vertex-hyb}{Sample-MH-Vertex-Hybrid}

In Algorithm \refalg{sample-MH-vertex-hyb} (\proc{sample-MH-vertex-hyb}) we use \proc{sample-star-vertex-IS} to sample a ``star'' vertex and then use the degree oracle to make sure its degree is greater than $k_2 = \sqrt{\tilde{m}}$. Then we normalize the return probability by canceling the starness factor of $u$. Note that in one procedure (\proc{sample-MHMH-edge-hyb}) we sample two MH-vertices without using this sampler, to reduce the penalty from square-logarithmic to logarithmic.

\begin{proc-algo}{sample-MH-vertex-hyb}{G,n,\eps;\tilde{m}}
    \alginput{$\tilde{m} \ge \max\{4, e^{-1/10} m\}$}
    \algoutput{For every vertex $v \in V$ with $\deg(u) > k_2$, the probability to return it is $\frac{e^{\pm \eps}}{300 \sqrt{\tilde{m}} \ln (\paperZcoefZestZindZinv/\eps)}$}
    \algcomplexity{$O(\log n)$ (expected)}
    \algcomplexity{$O(\log n \log \eps^{-1})$ (worst-case)}
    \begin{code}
        \algitem Let $k_2 \gets \sqrt{\tilde{m}}$.
        \algitem Let $\bu \gets \proc{sample-star-vertex-IS}(G,n,\eps;\tilde{m})$.
        \begin{If}{$\bu$ is \reject}
            \algitem Return \reject.
        \end{If}
        \begin{If}{$\bu \ne \reject$ and $\oracleDEG(\bu)\le  k_2$}
            \algitem Return \reject.
        \end{If}
        \algitem Filter by $\proc{estimate-indicator-inv}(\proc{starness-event-IS}(G,n;\tilde{m},\bu),\eps,1/30)$.
        \algitem Return $\bu$.
    \end{code}
\end{proc-algo}

\begin{lemma}{sample-MH-vertex-hyb}
    If $\tilde{m} \ge \max\{4, e^{-1/10}m\}$, then Algorithm \refalg{sample-MH-vertex-hyb} (\proc*{sample-MH-vertex-hyb}) samples every vertex whose degree is greater than $k_2$ with probability $e^{\pm \eps} / 300 \sqrt{\tilde{m}} \ln (\paperZcoefZestZindZinv/\eps)$, and otherwise rejects. Also, regardless of $\tilde{m}$, the query complexity is $O(\log n)$ in expectation and $O(\log n \log \eps^{-1})$ in worst-case.
\end{lemma}
\begin{proof}
    For complexity, observe that there is at most one $\oracleDEG$-call and one call to \proc{sample-star-vertex-IS}, which costs $O(\log n)$ $\oracleIS$-calls at worst-case (Lemma \reflemma{sample-star-vertex-IS}).

    If we reach the call to \proc{estimate-indicator-inv}, then $\deg(\bu) > k_2 = \sqrt{\tilde{m}}$. In this case, $\mathcal S_u(\tilde{m}) \ge 1/30$ (Lemma \reflemma{lbnd-starness}), and therefore, this call costs $O(1)$ calls to \proc{starness-event-IS} in expectation and $O(\log \eps^{-1})$ at worst-case (Lemma \reflemma{estimate-indicator-inv}). The cost for every such call is $O(\log n)$ (Observation \refobs{starness-event-IS}).

    Clearly, only MH-vertices can be returned. By Lemma \reflemma{sample-star-vertex-IS}, every MH-vertex $u$ is returned with probability $\frac{1}{10\sqrt{\tilde{m}}}\mathcal S_u(\tilde{m})$. The probability to pass the filter is in the range $e^{\pm \eps} \cdot \frac{1/30}{\mathcal S_u(\tilde{m}) \ln (\paperZcoefZestZindZinv/\eps)}$ (Lemma \reflemma{estimate-indicator-inv}). Combined, the probability to return any individual MH-vertex $u$ is in the range
    \[  \frac{1}{10\sqrt{\tilde{m}}} \mathcal S_u(\tilde{m}) \cdot e^{\pm \eps} \frac{1/30}{\mathcal S_u(\tilde{m}) \ln (\paperZcoefZestZindZinv/\eps)}
        = \frac{e^{\pm \eps}}{300\sqrt{\tilde{m}} \ln (\paperZcoefZestZindZinv/\eps)} \qedhere \]
\end{proof}

\subsection{Sampling L-MH edges (for low \texorpdfstring{$\tilde{m}$}{m~})}
\defproc{sample-LMH-edge-hyb}{Sample-L-MH-Edge-Hybrid}

This logic applies when $\tilde{m}$ is small enough for having $\sqrt{\tilde{m}} < \sqrt{n/\tilde{m}}$. In Algorithm \refalg{sample-LMH-edge-hyb} (\proc*{sample-LMH-edge-hyb}) we draw an MH-vertex $u$ and then a uniform neighbor $v$. To normalize the return probability for different choices of $u$, we normalize by a $\frac{\deg(u)}{2\tilde{m}}$-factor. When $\tilde{m}$ is higher, the result probability is too small (and too expensive to amplify), and we use different logic for medium-degree vertices and high-degree vertices.

\begin{proc-algo}{sample-LMH-edge-hyb}{G,n,\eps;\tilde{m}}
    \alginput{$\tilde{m} \ge \max\{4, e^{-1/10} m\}$}
    \algoutput{For every edge $e \in E_\mathrm{L,MH}$, the probability to return $e$ is $\frac{e^{\pm \eps}}{600 \tilde{m} \sqrt{\tilde{m}} \ln (\paperZcoefZestZindZinv/\eps)}$}
    \algcomplexity{$O(\log n)$ (expected)}
    \algcomplexity{$O(\log n \log \eps^{-1})$ (worst-case)}
    \begin{code}
        \algitem Let $k_2 \gets \sqrt{\tilde{m}}$.
        \algitem Let $u \gets \proc{sample-MH-vertex-hyb}(G,n,\eps;\tilde{m})$.
        \algitem Let $v \gets \oracleNEIGHBOR(u)$.
        \algitem Let $d_u \gets \oracleDEG(u)$.
        \begin{If}{$d_u > k_2$ and $\oracleDEG(v) \le k_2$}
            \algitem Return $uv$ with probability $\min\{1,\frac{d_u}{2\tilde{m}}\}$ (otherwise \reject).
        \end{If}
        \algitem Return \reject.
    \end{code}
\end{proc-algo}

\begin{lemma}{sample-LMH-edge-hyb}
    Assume that $\tilde{m} \ge \max\{4, e^{-1/10} m\}$. Algorithm \refalg{sample-LMH-edge-hyb} (\proc*{sample-LMH-edge-hyb}) samples every edge between vertices of degree of most $k_2=\sqrt{\tilde{m}}$ and vertices of degree greater than $k_2$ with probability $e^{\pm \eps} / 600 \tilde{m} \sqrt{\tilde{m}} \ln (\paperZcoefZestZindZinv/\eps)$. Moreover, regardless of $\tilde{m}$, the query complexity is $O(\log n)$ in expectation and $O(\log n + \log \eps^{-1})$ worst-case.
\end{lemma}

\begin{proof}
    For complexity, observe that there are at most two $\oracleDEG$-calls, one $\oracleNEIGHBOR$-call, and one call to \proc{sample-MH-vertex-hyb}, which costs $O(\log n)$ in expectation and $O(\log n + \log \eps^{-1})$ at worst-case (Lemma \reflemma{sample-MH-vertex-hyb}).

    Note that if $\tilde{m} \ge e^{-1/10} m$, then $d_u / 2\tilde{m} \le \frac{1}{2}e^{1/10} d_u / m \le 1$, where the last transition is correct since the number of edges $m$ cannot be smaller than the degree of $u$.

    Clearly, only L-MH edges can be returned. Let $u_0 v_0$ such an edge for which $\deg(u_0) > k_2$ and $\deg(v_0) \le k_2$:
    \begin{eqnarray*}
        \Pr\left[\text{return $u_0 v_0$}\right]
        &=& \Pr\left[\bu = u_0\right] \cdot \Pr\left[\bv = v_0 \cond \bu = u_0\right] \cdot \Pr\left[\text{pass filter}\cond \bu = u_0, \bv = v_0\right] \\
        \text{[Lemma \reflemma{sample-MH-vertex-hyb}]} &=& \frac{e^{\pm \eps}}{300 \sqrt{\tilde{m}} \ln (\paperZcoefZestZindZinv/\eps)} \cdot \frac{1}{\deg(u_0)} \cdot \frac{\deg(u_0)}{2\tilde{m}} \\
        &=& \frac{e^{\pm \eps}}{600 \tilde{m} \sqrt{\tilde{m}} \ln (\paperZcoefZestZindZinv/\eps)}
    \end{eqnarray*}
\end{proof}

\subsection{Sampling L-M edges (for high \texorpdfstring{$\tilde{m}$}{m~})}
\defproc{sample-LM-edge-hyb}{Sample-L-M-Edge-Hybrid}

This logic applies when $\tilde{m}$ is large enough for having $\sqrt{\tilde{m}} \ge \sqrt{n/\tilde{m}}$. In Algorithm \refalg{sample-LM-edge-hyb} (\proc{sample-LM-edge-hyb}) we use the same logic as in \proc{sample-LMH-edge-hyb}, but the upper bound for $\deg(u)$ is $k_3$ rather than $2\tilde{m}$.

\begin{proc-algo}{sample-LM-edge-hyb}{G,n,\eps;\tilde{m}}
    \alginput{$\tilde{m} \ge \max\{4, e^{-1/10}m\}$}
    \algoutput{For every edge $e \in E_\mathrm{L,M}$, the probability to return $e$ is $\frac{e^{\pm \eps}}{300 \tilde{m} \sqrt{n / \sqrt{\tilde{m}}} \ln (\paperZcoefZestZindZinv/\eps)}$}
    \algcomplexity{$O(\log n)$ (expected)}
    \algcomplexity{$O(\log n \log \eps^{-1})$ (worst-case)}
    \begin{code}
        \algitem Let $k_2 \gets \sqrt{\tilde{m}}$.
        \algitem Let $k_3 \gets \sqrt{n\sqrt{\tilde{m}}}$.
        \algitem Let $\bu \gets \proc{sample-MH-vertex-hyb}(G,n,\eps;\tilde{m})$.
        \algitem Let $\bv \gets \oracleNEIGHBOR(\bu)$.
        \algitem Let $d_{\bu} \gets \oracleDEG(\bu)$.
        \begin{If}{$k_2 < d_{\bu} \le k_3$ and $\oracleDEG(\bv) \le k_2$}
            \algitem Return $uv$ with probability $d_{\bu} / k_3$.
        \end{If}
        \algitem Return \reject.
    \end{code}
\end{proc-algo}

\begin{lemma}{sample-LM-edge-hyb}
    If $\tilde{m} \ge \max\{4, e^{-1/10} m\}$, then Algorithm \refalg{sample-LM-edge-hyb} (\proc*{sample-LM-edge-hyb}) samples every edge between vertices of degree of most $k_2$ and vertices of degree between $k_2$ and $k_3$ with probability in the range $e^{\pm \eps} / 300 \tilde{m} \sqrt{n / \sqrt{\tilde{m}}} \ln (\paperZcoefZestZindZinv/\eps)$. Also, regardless of $\tilde{m}$, the query complexity is $O(\log n)$ in expectation and $O(\log n \log \eps^{-1})$ in worst-case.
\end{lemma}
\begin{proof}
    For complexity, observe that there is at most two $\oracleDEG$-calls, one $\oracleNEIGHBOR$-call, and one call to \proc{sample-MH-vertex-hyb}, which costs $O(\log n)$ in expectation and $O(\log n \log \eps^{-1})$ at worst-case (Lemma \reflemma{sample-MH-vertex-hyb}).

    Clearly, only L-M edges can be returned. Let $u_0 v_0$ such an edge for which $k_2 < \deg(u_0) \le k_3$ and $\deg(v_0) \le k_2$:
    \begin{eqnarray*}
        \Pr\left[\text{return $u_0 v_0$}\right]
        &=& \Pr\left[\bu = u_0\right] \cdot \Pr\left[\bv = v_0 \cond \bu = u_0\right] \cdot \Pr\left[\text{pass filter}\cond \bu = u_0, \bv = v_0\right] \\
        \text{[Lemma \reflemma{sample-MH-vertex-hyb}]} &=& \frac{e^{\pm \eps}}{300 \sqrt{\tilde{m}} \ln (\paperZcoefZestZindZinv/\eps)} \cdot \frac{1}{\deg(u_0)} \cdot \frac{\deg(u_0)}{k_3} \\
        &=& \frac{e^{\pm \eps}}{300 \sqrt{\tilde{m}} k_3 \ln (\paperZcoefZestZindZinv/\eps)} \\
        &=& \frac{e^{\pm \eps}}{300 \tilde{m} \sqrt{n / \sqrt{\tilde{m}}} \ln (\paperZcoefZestZindZinv/\eps)} 
    \end{eqnarray*}
\end{proof}

\subsection{The tininess factor}
In this subsection we define the tininess factor which we use in the analysis of the L-H edge sampling logic (presented in the following subsection).

\defproc{tininess-event-LOCAL}{Tininess-Event-LOCAL}
We define the \emph{tininess} factor of a vertex $u$ with respect to the treshold degree $k_1(\sqrt{\tilde{m}}) = 3\sqrt{\tilde{m}^{3/2} / n}$ as the fraction of neighbors whose degree is bounded by $k_1$. Formally,
\[\mathcal T_u(\tilde{m}) = \frac{\abs{\{ v \in \mathrm{N}_u : \deg(v) \le k_1(\tilde{m}) \}}}{\deg(u)}\]

Algorithm \refalg{tininess-event-LOCAL} (\proc*{tininess-event-LOCAL}) draws a uniform neighbor of $u$ and compares its degree to $k_1(\tilde{m})$.

\begin{proc-algo}{tininess-event-LOCAL}{G,n; \tilde{m}, u}
    \algoutput{Binary answer. Accept probability $\mathcal T_u(\tilde{m})$}
    \algcomplexity{$O(1)$ worst-case}
    \begin{code}
        \algitem Let $k_1 \gets 3\sqrt{\tilde{m}^{3/2} / n}$.
        \algitem Let $\bv \gets \oracleNEIGHBOR(u)$.
        \algitem Let $d_{\bv} \gets \oracleDEG(\bv)$.
        \algitem Accept if and only if $d_{\bv} \le k_1$.
    \end{code}
\end{proc-algo}

\begin{observation}{tininess-event-LOCAL}
    Algorithm \refalg{tininess-event-LOCAL} (\proc*{tininess-event-LOCAL}) accepts the input $(G,n;\tilde{m},u)$ with probability exactly $\mathcal T_u(\tilde{m})$ at the cost of $O(1)$ local queries.
\end{observation}

The tininess factor of a high-degree vertex is $\Omega(1)$, as stated in the following lemma.

\begin{lemma}{lbnd-tininess}
    For a graph $G$ over $n$ vertices, recall that $k_1(\tilde{m}) = 3\sqrt{\tilde{m} \sqrt{\tilde{m}} / n}$ and $k_3(\tilde{m}) = \sqrt{n\sqrt{\tilde{m}}}$. If $\tilde{m} \ge e^{-1/10} m$, then for every vertex $u$ for which $\deg(u) > k_3(\tilde{m})$, $\mathcal T_u(\tilde{m}) \ge \frac{1}{4}$.
\end{lemma}
\begin{proof}
    The number of vertices of degree greater than $k_1$ in the graph is bounded by
    \begin{eqnarray*}
        \frac{2m}{k_1}
        = \frac{2m}{3 \sqrt{\tilde{m} \sqrt{\tilde{m}} / n}}
        = \frac{2}{3} \frac{m}{\tilde{m}} \sqrt{n \sqrt{\tilde{m}}}
        \le \frac{2e^{1/10}}{3} k_3
        \le \frac{3}{4} \deg(u)
    \end{eqnarray*}

    Hence, the number of $u$-neighbors whose degree is at most $k_1$ is at least $\frac{1}{4}\deg(u)$.
\end{proof}

\subsection{Sampling L-H edges (for high \texorpdfstring{$\tilde{m}$}{m~})}
\defproc{sample-LH-edge-LOCAL}{Sample-L-H-Edge-Local}

This logic applies when $\tilde{m}$ is large enough for having $\sqrt{\tilde{m}} \ge \sqrt{n/\tilde{m}}$.

Algorithm \refalg{sample-LH-edge-LOCAL} (\proc*{sample-LH-edge-LOCAL}) for sampling L-H edges is based on the concept of \cite{er18}. The main difference is the use of IS queries to cancel the tininess factor, reducing the query complexity dependency of $1/\eps$ from linear to logarithmic.

We uniformly choose a vertex $w$. Then, we draw a uniform neighbor $u$ (of $w$) and a uniform neighbor $v$ (of $u$). We reject immediately if $\deg(w) > k_1$ or if $\deg(u) \le k_3$ or if $\deg(v) > k_2$. Finally, we toss a coin to return the edge $uv$ with probability $\frac{\deg(w)}{k_1}$ and reject otherwise. We normalize the return probability by canceling the tininess factor, which is the fraction of $w$-neighbors of $u$ whose degree is bounded by $k_1$.

We consider vertices with degree $\deg(w) \le k_1$ as tiny-degree. Vertices of high degree have an $\Omega(1)$-fraction of tiny-degree neighbors, as stated in the following lemma. Recall that $\mathrm{N}_u$ denotes the set of $u$-neighbors.

\begin{proc-algo}{sample-LH-edge-LOCAL}{G,n,\eps;\tilde{m}}
    \alginput{$\tilde{m} \ge e^{-1/10}m$}
    \algoutput{For every edge $e \in E_\mathrm{L,H}$, the probability to return $e$ is $\frac{e^{\pm \eps}}{12 \tilde{m} \sqrt{n / \sqrt{\tilde{m}}} \ln (\paperZcoefZestZindZinv/\eps)}$}
    \algcomplexity{$O(1)$ (expected)}
    \algcomplexity{$O(\log \eps^{-1})$ (worst-case)}
    \begin{code}
        \algitem Let $k_1 \gets 3 \sqrt{\tilde{m} \sqrt{\tilde{m}} / n}$.
        \algitem Let $k_2 \gets \sqrt{\tilde{m}}$.
        \algitem Let $k_3 \gets \sqrt{n \sqrt{\tilde{m}}}$.
        \algitem Draw $\bw \in V$ uniformly.
        \algitem Let $\bu \gets \oracleNEIGHBOR(\bw)$.
        \algitem Let $\bv \gets \oracleNEIGHBOR(\bu)$.
        \algitem Let $d_{\bw} \gets \oracleDEG(\bw)$, $d_{\bu} \gets \oracleDEG(\bu)$, $d_{\bv} \gets \oracleDEG(\bv)$.
        \begin{If}{$d_{\bw} \le k_1$ and $d_{\bu} > k_3$ and $d_{\bv} \le k_2$}
            \algitem Proceed with probability $d_{\bw} / k_1$.
            \algitem Filter by $\proc{estimate-indicator-inv}(\proc{tininess-event-LOCAL}(G,n;\tilde{m},\bu), \eps, 1/4)$.
            \algitem Return $\bu\bv$.
        \end{If}
        \algitem Return \reject.
    \end{code}
\end{proc-algo}

\begin{lemma}{sample-LH-edge-LOCAL}
    Algorithm \refalg{sample-LH-edge-LOCAL} (\proc*{sample-LH-edge-LOCAL}) returns every edge $e \in E_{L,H}$ with probability $e^{\pm \eps} / 12 \tilde{m} \sqrt{n / \sqrt{\tilde{m}}} \ln (\paperZcoefZestZindZinv/\eps)$, and otherwise rejects. Moreover, the expected complexity is $O(1)$ and the worst-case complexity is $O(\log \eps^{-1})$.
\end{lemma}
\begin{proof}
    For complexity, observe that there are two $\oracleNEIGHBOR$-calls and three $\oracleDEG$-calls.

    If we reach the call to \proc{estimate-indicator-inv}, then $\deg(\bu) > k_3$. In this case, $\mathcal T_u(\tilde{m}) \ge \frac{1}{4}$, and therefore, this call costs $O(1)$ evaluations of \proc{tininess-event-LOCAL} in expectation, and $O(\log \eps^{-1})$ at worst-case (Lemma \reflemma{estimate-indicator-inv}). The cost for every such evaluation is $O(1)$ local queries (Observation \refobs{tininess-event-LOCAL}).

    Consider an edge $u_0 v_0$ where $\deg(u_0) > k_3$ and $\deg(v_0) \le k_2$. The probability to return the edge $u_0 v_0$ is:
    \begin{eqnarray*}
        \Pr\left[\text{return $u_0 v_0$}\right]
        &=& \sum_{w \in V} \Big( \Pr\left[\bw = w\right] \cdot \Pr\left[\bu = u_0 \cond \bw = w\right] \cdot \Pr\left[\bv = v_0 \cond \bu = u_0 \right] ~\cdots \\&& \phantom{\sum_{w\in V}}\cdots~ \Pr\left[\text{proceed} \cond \bu = u_0 \right] \cdot \Pr\left[\text{pass filter} \cond \bu = u_0 \right] \Big) \\
        &=& \sum_{\substack{w \in \mathrm{N}_{u_0} \\ \deg(w) \le k_1}} \frac{1}{n} \cdot \frac{1}{\deg(w)} \cdot \frac{1}{\deg(u_0)} \cdot \frac{\deg(w)}{k_1} \cdot \frac{e^{\pm\eps} \cdot 1/4}{\mathcal T_{u_0}(\tilde{m})\ln(\paperZcoefZestZindZinv/\eps)} \\
        &=& \frac{1}{n k_1 \deg(u_0)} \sum_{\substack{w \in \mathrm{N}_{u_0} \\ \deg(w) \le k_1}} \frac{e^{\pm\eps} \cdot 1/4}{\mathcal T_{u_0}(\tilde{m})\ln(\paperZcoefZestZindZinv/\eps)} \\
        &=& \frac{1}{4 n k_1 \underbrace{\deg(u_0)} \ln(\paperZcoefZestZindZinv/\eps)} \cdot \frac{1}{\mathcal T_{u_0}(\tilde{m})} \cdot e^{\pm\eps} \cdot \underbrace{\abs{\{ w \in \mathrm{N}_{u_0} : \deg(w) \le k_1 \}}}
    \end{eqnarray*}

    Recall that $\mathcal T_{u_0}(\tilde{m})$ is the fraction of tiny-degree neighbors of $u_0$ (according to the threshold $k_1$). Therefore,
    \[  \Pr\left[\text{return $u_0 v_0$}\right]
        = \frac{e^{\pm \eps}}{4 n k_1\ln(\paperZcoefZestZindZinv/\eps)}
        = \frac{e^{\pm \eps}}{12 \tilde{m} \sqrt{n / \sqrt{\tilde{m}}} \ln (\paperZcoefZestZindZinv/\eps)}
        \qedhere \]
\end{proof}

\subsection{Sampling MH-MH edges}
\defproc{sample-MHMH-edge-hyb}{Sample-MH-MH-Edge-Hybrid}

Algorithm \refalg{sample-MHMH-edge-hyb} (\proc*{sample-MHMH-edge-hyb}) merely samples two star vertices (plus normalization) and, if they both exist, tests whether there is an edge between them.

\begin{proc-algo}{sample-MHMH-edge-hyb}{G,n,\eps;\tilde{m}}
    \alginput{$\tilde{m} \ge e^{-1/10} m$}
    \algoutput{For every edge $e \in E_\mathrm{MH,MH}$, the probability to return $e$ is $\frac{e^{\pm \eps}}{45000 \tilde{m} \ln (\paperZcoefZestZindZinv/\eps)}$}
    \algcomplexity{$O(\log n)$ (expected)}
    \algcomplexity{$O(\log n \log \eps^{-1})$ (worst-case)}
    \begin{code}
        \algitem Let $k_2 \gets \sqrt{\tilde{m}}$.
        \algitem Let $\bu \gets \proc{sample-star-vertex-IS}(G,n;\tilde{m})$.
        \algitem Let $\bv \gets \proc{sample-star-vertex-IS}(G,n,\eps;\tilde{m})$.
        \begin{If}{$\bu,\bv \ne \reject$ and $\oracleIS(\{\bu,\bv\})$ rejects and $\oracleDEG(\bu) > k_2$ and $\oracleDEG(\bv) > k_2$}
            \algitem Filter by $\proc{estimate-indicator-inv}(\mathcal A, \eps, 1/900)$ with the following procedure $\mathcal A$:
            \begin{Codeblock*}
                \algitem Let $\bb_1 \gets \proc{starness-event-IS}(G,n;\tilde{m},\bu)$.
                \algitem Let $\bb_2 \gets \proc{starness-event-IS}(G,n;\tilde{m},\bv)$.
                \algitem Accept if and only if both $\bb_1$ and $\bb_2$ are ``\accept''.
            \end{Codeblock*}
            \algitem Return $\bu\bv$.
        \end{If}
        \algitem Return \reject.
    \end{code}
\end{proc-algo}

\begin{lemma}{sample-MHMH-edge-hyb}
    If $\tilde{m} \ge \max\{4, e^{-1/10} m\}$, then Algorithm \refalg{sample-MHMH-edge-hyb} (\proc*{sample-MHMH-edge-hyb}) returns every edge $e \in E_{MH,MH}$ with probability $e^{\pm \eps} / 45000 \tilde{m} \ln^2 (\paperZcoefZestZindZinv/\eps)$, and otherwise rejects. Moreover, regardless of $\tilde{m}$, the complexity is $O(\log n)$ in expectation if $\tilde{m} \ge e^{-1/10} m$ and $O(\log n \log \eps^{-1})$ in worst-case.
\end{lemma}
\begin{proof}
    For complexity, observe that there are three explicit oracle calls, in addition to two calls to \proc{sample-star-vertex-IS}, each costing $O(\log n)$ (Lemma \reflemma{sample-star-vertex-IS}).

    If we reach the call to \proc{estimate-indicator-inv}, then $\deg(\bu) > k_2 = \sqrt{\tilde{m}}$. In this case, $\mathcal S_u(\tilde{m}) \mathcal S_v(\tilde{m}) \ge (1/30)^2 = 1/900$ (Lemma \reflemma{lbnd-starness}), and therefore, this call for costs $O(1)$ calls to \proc{starness-event-IS} in expectation and $O(\log \eps^{-1})$ at worst-case (Lemma \reflemma{estimate-indicator-inv}). The cost for every such call is $O(\log n)$ (Observation \refobs{starness-event-IS}).
    
    Consider a directed edge $u_0 \to v_0$ for which $\deg(u_0), \deg (v_0) > k_2$. The probability to return $u_0 \to v_0$ is
    \begin{eqnarray*}
        \Pr\left[\text{return $u_0 \to v_0$}\right]
        &=&
        \Pr\left[\bu = u_0\right] \cdot \Pr\left[\bv = v_0\right] \cdot \Pr\left[\text{pass filter}\cond \bu=u_0, \bv=v_0\right] \\
        &=& \frac{\mathcal S_{u_0}(\tilde{m})}{10\sqrt{\tilde{m}}} \cdot \frac{\mathcal S_{v_0}(\tilde{m})}{10\sqrt{\tilde{m}}} \cdot e^{\pm \eps} \frac{1/900}{\mathcal S_{u_0}(\tilde{m}) \mathcal S_{v_0}(\tilde{m}) \cdot \ln (\paperZcoefZestZindZinv / \eps)} \\
        &=& \frac{e^{\pm \eps}}{90000 \tilde{m} \ln (\paperZcoefZestZindZinv / \eps)}
    \end{eqnarray*}

    Note that the probability to return the edge $u_0 v_0$ is twice this result, since it can also be represented as $v_0 u_0$.
\end{proof}

\subsection{Sampling all edges}
\defproc{sample-edge-core-hyb}{Sample-Edge-Core-Hybrid}

We combine the category-specific edge sampling procedures by unifying their coefficient.

\begin{proc-algo}{sample-edge-core-hyb}{G,n,\eps;\tilde{m}}
    \alginput{$\tilde{m} \ge \max\{4, e^{-1/10} m\}$}
    \algoutput[noperiod]{For every edge $e \in E$, the probability to return $e$ is $\frac{e^{\pm 2\eps}}{750000 \tilde{m} R \ln (\paperZcoefZestZindZinv/\eps)}$}
    \algoutputphantom{~~~}{where $R = \min\left\{\sqrt{\tilde{m}}, \sqrt{n/\sqrt{\tilde{m}}}\right\}$}
    \begin{code}
        \algitem Let $\tilde{R} \gets \min\{\sqrt{\tilde{m}}, \sqrt{n/\sqrt{\tilde{m}}}\}$.
        \algitem Let $e \gets \mathrm{N/A}$.
        \begin{If}{$\sqrt{\tilde{m}} \le \sqrt{n/\sqrt{\tilde{m}}}$}
                \algitem Toss a four-head coin.
                \begin{Codeblock*}
                    \algitem With probability $\frac{1}{225 \tilde{R}}$: let $e \gets \proc{sample-LL-edge-hyb}(G,n,\eps;\tilde{m})$.
                    \algitem With probability $\frac{1}{75}$: let $e \gets \proc{sample-LMH-edge-hyb}(G,n,\eps;\tilde{m})$.
                    \algitem With probability $\frac{1}{\tilde{R}}$: let $e \gets \proc{sample-MHMH-edge-hyb}(G,n,\eps;\tilde{m})$.
                    \algitem Otherwise: (do nothing).
                \end{Codeblock*}
        \end{If}
        \begin{Else}
            \algitem Toss a five-head coin.
            \begin{Codeblock*}
                \algitem With probability $\frac{1}{225 \tilde{R}}$: let $e \gets \proc{sample-LL-edge-hyb}(G,n,\eps;\tilde{m})$.
                \algitem With probability $\frac{1}{150}$: let $e \gets \proc{sample-LM-edge-hyb}(G,n,\eps;\tilde{m})$.
                \algitem With probability $\frac{1}{3750}$: let $e \gets \proc{sample-LH-edge-LOCAL}(G,n,\eps;\tilde{m})$.
                \algitem With probability $\frac{1}{\tilde{R}}$: let $e \gets \proc{sample-MHMH-edge-hyb}(G,n,\eps;\tilde{m})$.
                \algitem Otherwise: (do nothing).
            \end{Codeblock*}
        \end{Else}
        \begin{If}{$e \ne \mathrm{N/A}$}
            \algitem Return $e$.
        \end{If}
        \algitem Return \reject.
    \end{code}
\end{proc-algo}

\begin{lemma}{sample-edge-core-hyb}
    For every graph $G = (V,E)$ (over $n$ vertices and $m$ edges), if $\tilde{m} \ge \max\{4, e^{-1/10} m\}$, then Algorithm \refalg{sample-edge-core-hyb} (\proc*{sample-edge-core-hyb}) returns every edge $uv \in E$ with probability in the range $e^{\pm 2 \eps} / 45000\tilde{m} \tilde{R} \ln (\paperZcoefZestZindZinv/\eps)$ for $\tilde{R} = \min\{\sqrt{\tilde{m}}, \sqrt{n / \sqrt{\tilde{m}}}\}$, and otherwise rejects. The query complexity is $O(\log n)$ in expectation and $O(\log n \log \eps^{-1})$ in worst-case.
\end{lemma}
\begin{proof}
    For query complexity, observe that all subroutines require at most $O(\log n)$ queries in expectation and $O(\log n \log \eps^{-1})$ queries at worst-case.
    
    If $\sqrt{\tilde{m}} \le \sqrt{n / \sqrt{\tilde{m}}}$ (that is, $\tilde{R} = \sqrt{\tilde{m}}$), then the probability to return every individual edge $uv$, based on its vertices' degree classes, is:
    \[ \begin{array}{lllll}
         \text{Edge class} & \text{Lemma} & \text{Base probability} & \text{Coin} & \text{Combined} \\
         \text{L-L} & \text{\reflemma{sample-LL-edge-hyb}} & \frac{e^{\pm \eps}}{200\tilde{m} \ln (\paperZcoefZestZindZinv/\eps)} & \frac{1}{225 \tilde{R}} & \frac{e^{\pm \eps}}{45000 \tilde{m} \tilde{R} \ln (\paperZcoefZestZindZinv/\eps)} \\
         \text{L-MH} & \text{\reflemma{sample-LMH-edge-hyb}} & \frac{e^{\pm \eps}}{600\tilde{m} \tilde{R} \ln (\paperZcoefZestZindZinv/\eps)} & \frac{1}{75} & \frac{e^{\pm \eps}}{45000 \tilde{m} \tilde{R} \ln (\paperZcoefZestZindZinv/\eps)} \\
         \text{MH-MH} & \text{\reflemma{sample-MHMH-edge-hyb}} & \frac{e^{\pm \eps}}{45000\tilde{m} \ln (\paperZcoefZestZindZinv/\eps)} & \frac{1}{\tilde{R}} & \frac{e^{\pm \eps}}{45000 \tilde{m} \tilde{R} \ln (\paperZcoefZestZindZinv/\eps)}
    \end{array} \]
    
    If $\sqrt{\tilde{m}} > \sqrt{n / \sqrt{\tilde{m}}}$ (that is, $\tilde{R} = \sqrt{n / \sqrt{\tilde{m}}}$), then the probability to return every individual edge $uv$, based on its vertices' degree classes, is:
    \[ \begin{array}{lllll}
        \text{Edge class} & \text{Lemma} & \text{Base probability} & \text{Coin} & \text{Combined} \\
        \text{L-L} & \text{\reflemma{sample-LL-edge-hyb}} & \frac{e^{\pm \eps}}{200\tilde{m} \ln (\paperZcoefZestZindZinv/\eps)} & \frac{1}{225 \tilde{R}} & \frac{e^{\pm \eps}}{45000 \tilde{m} \tilde{R} \ln (\paperZcoefZestZindZinv/\eps)} \\
        \text{L-M} & \text{\reflemma{sample-LM-edge-hyb}} & \frac{e^{\pm \eps}}{300\tilde{m} \tilde{R} \ln (\paperZcoefZestZindZinv/\eps)} & \frac{1}{150} & \frac{e^{\pm \eps}}{45000 \tilde{m} \tilde{R} \ln (\paperZcoefZestZindZinv/\eps)} \\
        \text{L-H} & \text{\reflemma{sample-LH-edge-LOCAL}} & \frac{e^{\pm \eps}}{12\tilde{m} \tilde{R} \ln (\paperZcoefZestZindZinv/\eps)} & \frac{1}{3750} & \frac{e^{\pm \eps}}{45000 \tilde{m} \tilde{R} \ln (\paperZcoefZestZindZinv/\eps)} \\
        \text{MH-MH} & \text{\reflemma{sample-MHMH-edge-hyb}} & \frac{e^{\pm \eps}}{45000\tilde{m} \ln (\paperZcoefZestZindZinv/\eps)} & \frac{1}{\tilde{R}} & \frac{e^{\pm \eps}}{45000 \tilde{m} \tilde{R} \ln (\paperZcoefZestZindZinv/\eps)}
    \end{array} \]
\end{proof}

\defproc{sample-edge-amplified-hyb}{Sample-Edge-Amplified-Hybrid}
The core sampler is a $(\lambda,\eps)$-uniform sampler, but $\lambda$ can be very small. Algorithm \refalg{sample-edge-amplified-hyb} (\proc*{sample-edge-amplified-hyb}) amplifies the success probability to $3/4$ by making $O(R \log (1/\eps))$ sample tries, for $R = \min\{\sqrt{m}, \sqrt{n / \sqrt{m}}\}$, and choosing the first successful one. This algorithm requires an advice $\tilde{m}$, and its correctness and expected query complexity statements only hold if $\tilde{m} \in e^{\pm 1/10} m$.

\begin{proc-algo}{sample-edge-amplified-hyb}{G, n, \tilde{m}, \eps}
    \alginput{$\tilde{m} \ge \max\{4, e^{-1/10} m\}$}
    \algoutput{For every edge $e \in E$, the probability to return $e$ is $e^{\pm \eps} \lambda(G,\tilde{m})$}
    \algoutputphantom{}{($\lambda(G,\tilde{m}) \ge 3/4$ if $\tilde{m} \ge e^{-1/10} m$)}
    \algcomplexity{$O(\min\{\sqrt{m}, n / \sqrt{m}\} \cdot \log n \log (1/\eps))$ expected, if $\tilde{m} \ge e^{-1/10} m$}
    \algcomplexity{$O(n^{1/3} \log n \log^2 (1/\eps))$ worst-case}
    \begin{code}
        \algitem Let $\tilde{R} \gets \min\{\sqrt{\tilde{m}}, \sqrt{n / \sqrt{\tilde{m}}}\}$.
        \begin{For}{$\ceil{10^6 \tilde{R} \ln (\paperZcoefZestZindZinv/\eps)}$ times}
            \algitem Let $\be \gets \proc{sample-edge-core-hyb}(G,n,\eps/5;\tilde{m})$.
            \begin{If}{$\be$ is not \reject}
                \algitem Return $\be$.
            \end{If}
        \end{For}
        \algitem Return \reject.
    \end{code}
\end{proc-algo}

\begin{lemma}{sample-edge-amplified-hyb}
    For every graph $G = (V,E)$ (over $n$ vertices and $m$ edges), if $\tilde{m} \ge \max\{4, e^{-1/10} m\}$, then Algorithm \refalg{sample-edge-core-hyb} (\proc*{sample-edge-core-hyb}) returns every edge $uv \in E$ with probability in the range $e^{\pm \eps} \lambda(G,\tilde{m})$, where $\lambda(G,\tilde{m}) \ge 3/4$, and otherwise rejects. The query complexity is $O(R \log n \log (1/\eps))$ in expectation (based on the assumption about $\tilde{m}$) for $R = \min\{\sqrt{\tilde{m}}, \sqrt{n / \sqrt{\tilde{m}}}\}$ and $O(n^{1/3} \log n \cdot \log^2 (1/\eps))$ in worst-case (for every $\tilde{m}$).
\end{lemma}
\begin{proof}
    Let $R = \min\{\sqrt{m}, \sqrt{n / \sqrt{m}}\}$ and $\tilde{R} = \min\{\sqrt{\tilde{m}}, \sqrt{n / \sqrt{\tilde{m}}}\}$. The success probability of an iteration is at least $e^{-\eps/5} / 45000 \tilde{m} \tilde{R} \ln (\paperZcoefZestZindZinv/\eps)$ (Lemma \reflemma{sample-edge-core-hyb}). Therefore, the success probability of the loop is:
    \[  \lambda(G,\tilde{m})
        \ge 1 - \left(1 - m \cdot \frac{e^{-\eps/5}}{45000 \tilde{m} \tilde{R} \ln (\paperZcoefZestZindZinv/\eps)}\right)^{10^6 \tilde{R} \ln (\paperZcoefZestZindZinv/\eps)}
        \ge 1 - e^{-e^{-1/10} \cdot e^{-1/5} \cdot \frac{200}{9}}
        \ge 3/4
    \]

    For every edge $e \in E$, the probability to sample it is in the range $e^{\pm 2 (\eps/5)} \lambda(G, \tilde{m})$.

    For complexity:
    \begin{itemize}
        \item Expected: $O(\tilde{R} \log (1/\eps))$ rounds times $O(\log n)$ per-round (Lemma \reflemma{sample-edge-core-hyb}), which is $O(R \log n \cdot \log (1/\eps))$ if $\tilde{m} \in e^{\pm 1/10} m$.
        \item Worst-case: $O(n^{1/3} \log (1/\eps))$ rounds times $O(\log n \log (1/\eps))$ per-round (Lemma \reflemma{sample-edge-core-hyb}).
    \end{itemize}
\end{proof}

\defproc{sample-edge-hyb}{Sample-Edge-Hybrid}
The entry-point sampling algorithm, \proc*{sample-edge-hyb} (Algorithm \refalg{sample-edge-hyb}), estimates $\tilde{m} \in e^{\pm 1/10} m$ with success probability $1 - r$, where $r = O(\eps^2) \cap O(1/n^{1/3} \log n)$.

\begin{proc-algo}{sample-edge-hyb}{G,n,\eps}
    \begin{code}
        \begin{If}{$\eps > 1/3$}
            \algitem Set $\eps \gets 1/3$.
        \end{If}
        \algitem Let $r \gets \min\{\eps^2 / 6, 1/n^2 \log n \log^2 (1/\eps) \}$.
        \algitem Compute $\tilde{\bm} \gets e^{\pm 1/10} m$ with probability $1 - r$.
        \begin{If}{$\tilde{\bm} < 4$}
            \algitem Return $\proc{sample-edge-bruteforce-IS}(G)$.
        \end{If}
        \algitem Return $\proc{sample-edge-amplified-hyb}(G,n,\eps/5,\tilde{\bm})$.
    \end{code}
\end{proc-algo}

\begin{lemma}{sample-edge-hyb}
    For an input graph $G$ over $n$ vertices (known to the algorithm) and $m$ edges (unknown to the algorithm), Algorithm \refalg{sample-edge-hyb} (\proc*{sample-edge-hyb}) is a $(2/3, \eps)$-uniform sampler of the edges of $G$. The expected query complexity is $O(R \log n \log (n/\eps) + \log^2 (1/\eps))$, where $R=\min\{\sqrt{m}, \sqrt{n/\sqrt{m}}\}$.
\end{lemma}
\begin{proof}
    Estimating $\tilde{m}$ in the range $e^{\pm 1/10} m$ with probability $2/3$ requires $O(R \log n)$ queries worst-case (Lemma \reflemma{external:estimate-edges-hybrid}). To increase this probability to $1 - r$, for $r=\min\{\eps^2/6,1/n^2\log n \log^2 (1/\eps)\}$, we take the median of $O(\log r^{-1}) = O(\log (n/\eps))$ independent runs at the cost of $O(R \log (n/\eps) \log n)$ queries.
    
    The contribution of the brute-force sampler in the case where $\tilde{\bm}$ is wrongly small is $\Pr\left[\tilde{\bm} < 4\right] \cdot O(1 + m \log n) = O(r m \log n) = O(1)$ (last transition is correct since $m \le n^2$ and $r \le 1/n^2 \log n$).

    If $\tilde{m} \ge 4$, then amplified sample costs:
    \begin{itemize}
        \item $O(R \cdot \log n)$ in expectation, if $\tilde{m} \in e^{\pm 1/10} m$.
        \item $O(n^{1/3} \log n \log^2 (1/\eps))$ if $\tilde{m} \notin e^{\pm 1/10} m$
    \end{itemize}
    Since $\tilde{m}$ is outside the correct range with probability at most $r \le 1/n^2 \log n \log^2 (1/\eps)$, the expected query complexity of the sampling logic is $O(R \log n \log (1/\eps))$.
    
    Let $\lambda'(G) = \frac{\sum_{a \in e^{\pm 1/10} m} \Pr\left[\tilde{m} = a\right]}{\Pr\left[\tilde{m} \in e^{\pm 1/10} m\right]} \ge \frac{3}{4}$ (last transition: Lemma \reflemma{sample-edge-amplified-hyb}).

    By Lemma \reflemma{sample-edge-amplified-hyb}, the probability to sample an individual edge $e \in E$ is at least $e^{-\eps/5} \lambda'(G)$ and at most $e^{\eps/5}\lambda'(G) + r$. For $0 < \eps \le 1/3$, the ratio between the maximum and the minimum probability is at most:
    \begin{eqnarray*}
        e^{(2/5)\eps} + \frac{e^{\eps/5}}{\lambda'(G)} r
        \le e^{(2/5)\eps} + \frac{4}{3} e^{\eps/5} \cdot \frac{1}{6} \eps^2
        \le e^{\eps/2}
    \end{eqnarray*}

    Therefore, the output is $\eps$-uniform when conditioned on success.
\end{proof}

\section{Using the IS oracle}
\label{sec:IS-only}

The main conceptual difference between the IS-only algorithm and the IS+LOCAL algorithm is the inability to certainly determine the category of each vertex, since we have to estimate the degree rather than using the degree oracle. Besides this, the algorithms differ by the need to simulate the local oracles using the IS-oracle.

In the IS-only part, we use the following parameters:
\begin{itemize}
    \item $k = \sqrt{\tilde{m}}$ - a threshold degree. This is the same as $k_2$ of the technical overview.
    \item $n^* = \min\{n, 2m\}$ - an upper bound for the number of effective vertices.
    \item $\tilde{n}^* = \min\{n, 2\tilde{m}\}$ - an estimation for $n^*$.
\end{itemize}
Degrees bounded by $k$ are considered low and degrees greater than $k$ are considered as high. Some statements refer to an extended range for low degrees, between $0$ and $2k$ (rather than $k$). This follows the need to test the degree of a vertex instead of certainly determine it using the degree oracle.

Let $r = \poly(\eps, 1/n)$ be a reasonable additive error. The algorithm uses $O(\min\{\sqrt{m}, \sqrt{n/\sqrt{m}}\} \cdot \poly(\log n) \cdot \log (1/r))$ queries to obtain an approximation $\tilde{m}$, which is in the range $e^{\pm 1/10} m$ with probability at least $1 - r$ \cite{clw20, ahl26}.

\subsection{Sampling a neighbor}
\defproc{sample-neighbor-IS}{Sample-Neighbor-IS}
The neighbor sampler for parameters $n$, $u$ and $\tilde{m}$ has two branches, depending on the choice of $\tilde{n}^*$.
\begin{itemize}
    \item If $\tilde{n}^* = n$, then we uniformly choose a vertex $v$ and use a single IS query to determine whether $v \in \mathrm{N}_u$, in which case we return it, or not, in which case we reject.
    \item If $\tilde{n}^* \ne n$, then we draw a set $\bS$ of density $1/\tilde{m}$. If there is exactly one edge in $\bS \cup \{u\}$, and this edge is adjacent to $u$, then we return its other endpoint vertex. Otherwise we reject.
\end{itemize}

Algorithm \refalg{sample-neighbor-IS} (\proc*{sample-neighbor-IS}) provides the pseudocode for the neighbor sampler. In the following we define $\mathcal N_{u,v}(\tilde{m})$, the \emph{neighborhood factor} of $u$ and $v$, and show that the probability to sample $v$ (when $u$ is given) is exactly $\mathcal N_{u,v}(\tilde{m})/\tilde{m}$. This can be see as an alternative to the loneliness factor, when considering a different density.

\begin{proc-algo}{sample-neighbor-IS}{G,n,\tilde{m};u}
    \alginput{$\tilde{m} \ge \max\{5, e^{-1/10} m\}$}
    \alginput{A vertex $u$}
    \alginput{If $\tilde{n}^* = n$: for every $v \in \mathrm{N}_u$, the probability to return $v$ is $\frac{1}{n}$}
    \alginput{If $\tilde{n}^* \ne n$: for every $v \in \mathrm{N}_u$, the probability to return $v$ is $\frac{1}{\tilde{m}} \mathcal N_{u,v}(\tilde{m})$}
    \algcomplexity{$1$ if $\tilde{n}^* = n$ (worst-case)}
    \algcomplexity{$O(\log n)$ if $\tilde{n}^* \ne n$ (worst-case)}
    \begin{code}
        \algitem Let $\tilde{n}^* \gets \min\{n, 2\tilde{m}\}$.
        \begin{If}{$\tilde{n}^* = n$}
            \algitem Draw $\bv \in V$ uniformly.
            \algpushcomment{(No edges)}
            \begin{If}{$\oracleIS(\{u,\bv\})$ accepts}
                \algitem Return \reject.
            \end{If}
            \algitem Return $\bv$.
        \end{If}
        \begin{Else}
            \algitem Let $\bS \subseteq V$ be a set that every non-$u$ vertex belongs to with probability $1/\tilde{m}$ iid.
            \algitem Let $\be \gets \proc{extract-edge-IS}(G; \bS \cup \{u\})$.
            \begin{If}{$\be \ne \reject$ and $u \in \be$}
                \algitem Let $\bv$ be the other vertex of $\be$.
                \begin{If}{$\proc{test-loneliness-IS}(\bS; u, \bv)$}
                    \algitem Return $\bv$.
                \end{If}
            \end{If}
            \algitem Return \reject.
        \end{Else}
    \end{code}
\end{proc-algo}

\defproc{neighborhood-event-IS}{Neighborhood-Event-IS}
To draw an indicator whose expected value is $\mathcal N_{u,v}(\tilde{m})$, we draw a set $\bS \subseteq V$ of density $p = 1/\tilde{m}$ and test whether or not $\{u,v\}$ is lonely in $\bS$. Algorithm \refalg{neighborhood-event-IS} (\proc*{neighborhood-event-IS}) provides the pseudocode for this logic.

\begin{proc-algo}{neighborhood-event-IS}{G,n,\tilde{m};u,v}
    \alginput{$\tilde{m} \ge \max\{5, e^{-1/10}m\}$}
    \alginput{A vertex $u$, a vertex $v \in \mathrm{N}_u$}
    \algoutput{Expected value: $\mathcal N_{u,v}(\tilde{m})$}
    \algcomplexity{$O(1)$ (worst case)}
    \begin{code}
        \algitem Draw a set $\bS \subseteq V$ of density $1/\tilde{m}$.
        \algitem Return $\proc{test-loneliness-IS}(\bS; u,v)$.
    \end{code}
\end{proc-algo}

\begin{observation}{neighborhood-event-IS}
    Algorithm \refalg{neighborhood-event-IS} (\proc*{neighborhood-event-IS}) accepts the input $(G,n,\tilde{m};u,v)$ with probability exactly $\mathcal N_{u,v}(\tilde{m})$. The query complexity is $O(1)$ worst-case.
\end{observation}

The neighborhood factor of a high degree vertex is $\Omega(1)$, as stated in the following lemma.

\begin{lemma}{lbnd-neighborhood}
    Recall that $k = \sqrt{\tilde{m}}$ and assume that $\tilde{m} \ge \max\{20, e^{-1/10} m\}$. For every vertex $u$ and a neighbor $v \in \mathrm{N}_u$, $\mathcal N_{u,v}(\tilde{m}) \ge \frac{1}{25}$.
\end{lemma}
\begin{proof}
    The following proof is identical to the proof of Lemma \reflemma{lbnd-loneliness}, up to the different density parameter.

    The neighborhood event is the negation of the union of the following event:
    \begin{itemize}
        \item $\bS$ has a $u$-neighbor which is not $v$.
        \item $\bS$ has a $v$-neighbor which is not $u$.
        \item $\bS$ has an edge adjacent to neither $u$ nor $v$.
    \end{itemize}
    
    The probability to have a $u$-neighbor (which is not $v$) or a $v$-neighbor (which is not $u$) is bounded by:
    \begin{eqnarray*}
        1 - \left(1 - p\right)^{(\deg(u) - 1) + (\deg(v) - 1)}
        &\!\!\le\!\!& 1 - \left(1 - p\right)^{2m} \\
        \text{[Lemma \reflemma{lbnd-1-plus-x-power-y}]}
        &\!\!\le\!\!& 1 - (1 - p^2 \cdot 2m)e^{-p \cdot 2m} \\
        &\!\!=\!\!& 1 - \left(1 - \frac{2m}{\tilde{m}^2}\right)e^{-2m/\tilde{m}} \\
        &\!\!\le\!\!& 1 - \left(1 - \frac{2e^{1/10}}{\tilde{m}}\right)e^{-2e^{1/10}}
        \le 1 - \left(1 - \frac{2e^{1/10}}{20}\right)e^{-2e^{1/10}}
        \le 0.903
    \end{eqnarray*}
    
    The probability to have a non-$u$, non-$v$ edge is bounded by $p^2 m = m/\tilde{m}^2 \le \frac{e^{1/10}}{\tilde{m}} \le \frac{e^{1/10}}{20} \le 0.056$.

    Combined, $\mathcal N_{u,v}(\tilde{m}) \ge 1 - 0.903 - 0.056 = 0.041 > \frac{1}{25}$.
\end{proof}

\begin{lemma}{sample-neighbor-IS}
    A call to $\proc*{sample-neighbor-IS}(G,n;\tilde{m},u)$ samples every $v \in \mathrm{N}_u$ with probability exactly $\frac{1}{n}$ if $\tilde{n}^* = n$ and otherwise $\frac{1}{\tilde{m}} \mathcal N_{u,v}(\tilde{m})$, at the cost of $O(\log n)$ IS-queries worst-case.
\end{lemma}
\begin{proof}
    For complexity: if $\tilde{n}^* = n$, then we make a single IS query. Otherwise, the call of \proc{extract-edge-IS} costs $O(\log n)$ worst-case (Lemma \reflemma{extract-edge-IS}) and the call of \proc{test-loneliness-IS} is $O(1)$ worst-case (Observation \refobs{test-loneliness-IS}).

    Correctness is trivial if $\tilde{n}^* = n$. If $\tilde{n}^* \ne n$, then:
    \[  \Pr\left[\text{sample $v$} \cond u \right]
        = \Pr\left[v \in \bS \right]\Pr\left[\text{$uv$ is lonely in $\bS$} \cond v \in \bS \right]
        = \frac{1}{\tilde{m}} \cdot \mathcal N_{u,v}(\tilde{m})
    \]
\end{proof}

\subsection{Categorizing the degree}
\defproc{test-high-degree-IS}{Test-High-Degree-IS}

Recall the threshold degree $k=\sqrt{\tilde{m}}$ and the tolerable additive error $r$. For every vertex $u$, we test whether its degree is low or high. If $\deg(u) \le k$ then we say ``low'' with probability at least $1-r$, and if $\deg(u) \ge 2k$ then we say ``high'' with probability at least $1-r$. The high degrees in the range between $k$ and $2k$ can be wrongly classified as ``low'' with any probability. The query complexity of the low-high test (Algorithm \refalg{test-high-degree-IS}) is $O(\log n \log (1/r))$ worst-case.

To distinguish between the degree threshold, we draw a set $\bS$ of density $1/8k$. The event ``$\bS \cup \{u\}$ has edges but $\bS \setminus \{u\}$ is independent'' is bounded by a constant probability if $\deg(u) \le k$ and greater than (another) constant probability if $\deg(u) \ge 2k$. We use $O(\log (1/r)) = O(\log (n/\eps))$ rounds to amplify the success probability of the test to $1-r$.

\begin{proc-algo}{test-high-degree-IS}{G,n,\tilde{m};u, r}
    \alginput{$\tilde{m} \ge \max\{36, e^{-1/10}m\}$}
    \algcomplexity{$O(\log (1/r))$ worst-case}
    \algcompleteness{If $\deg(u) \ge 2k$, then the accept probability is at least $1-r$}
    \algsoundness{If $\deg(u) \le k$, then the reject probability is at least $1-r$}
    \begin{code}
        \algitem Let $k \gets \sqrt{\tilde{m}}$.
        \algitem Let $N \gets \ceil{400 \ln (1/r)}$
        \algitem Set $\bM \gets 0$.
        \begin{For}{$N$ times}
            \algitem Draw a set $\bS \subseteq V$ that every non-$u$ element belongs to with probability $1/8k$ iid.
            \begin{If}{$\oracleIS(S \cup \{u\})$ rejects and $\oracleIS(S \setminus \{u\})$ accepts}
                \algitem Set $\bM \gets \bM + 1$.
            \end{If}
        \end{For}
        \begin{If}{$\bM \ge \frac{13}{80}N$}
            \algitem Return \accept.
        \end{If}
        \begin{Else}
            \algitem Return \reject.
        \end{Else}
    \end{code}
\end{proc-algo}

\begin{lemma}{test-high-degree-IS}
    Assume that $\tilde{m} \ge \min\{36,e^{-1/10}m\}$. Procedure $\proc*{test-high-degree-IS}(G,n,\tilde{m},\delta;u)$ (Algorithm \refalg{test-high-degree-IS}) accepts with probability at least $1-r$ if $\deg(u) \ge 2k$ and rejects with probability at least $1-r$ if $\deg(u) \le k$. Moreover, the query complexity is $O(\log (1/r))$ worst-case.
\end{lemma}
\begin{proof}
    For complexity, observe that we make $O(\log (1/r))$ rounds, each costing two IS queries.
    
    If $\deg(u) \le k$, then the probability that $\bS \cup \{u\}$ has edges but $\bS \setminus \{u\}$ is independent is bounded by the expected number of $u$-edges, which is $(1/8k) \cdot \deg(u) \le 1/8$.
    
    Therefore, $\bM$ distributes as $\Bin(N,p)$ for $p \le \frac{1}{8}$. The probability to accept is bounded by:
    \begin{eqnarray*}
        \Pr\left[\Bin\left(N,\frac{1}{8}\right) \ge \frac{13}{80}N\right]
        &=& \Pr\left[\Bin\left(N,\frac{1}{8}\right) - \frac{1}{8}N \ge \frac{3}{80} N\right] \\
        &\le& e^{-2\cdot((3/80)N)^2/N}
        \le e^{-(9/3200) N}
        \le e^{-\ln r^{-1}}
        = r
    \end{eqnarray*}
    
    If $\deg(u) \ge 2k$, then the probability that $\bS \cup \{u\}$ has edges but $\bS \setminus \{u\}$ is independent is at least:
    \begin{eqnarray*}
        \Pr\left[\bS \cap \mathrm{N}_u \ne \emptyset\right] -\E\left[\text{edges in $\bS$}\right]
        &\ge& \left(1 - \left(1 - \frac{1}{8k}\right)^{\deg(u)}\right) - \frac{1}{64k^2} \cdot m \\
        &\ge& (1 - e^{-(1/8k) \cdot (2k)}) - \frac{1}{64\tilde{m}} \cdot m
        \ge 1 - e^{-1/4} - \frac{e^{1/10}}{64}
        \ge \frac{1}{5}
    \end{eqnarray*}
    
    Therefore, $M$ distributes as $\Bin(N,p)$ for $p > \frac{1}{5}$. The probability to reject is bounded by:
    \begin{eqnarray*}
        \Pr\left[\Bin\left(N,\frac{1}{5}\right) < \frac{13}{80}N\right]
        &=& \Pr\left[\Bin\left(N,\frac{1}{5}\right) - \frac{1}{5}N < -\frac{3}{80}N\right] \\
        &\le& e^{-2((3/80)N)^2/N}
        \le e^{-(9/3200)N}
        \le e^{-\ln r^{-1}}
        = r
    \end{eqnarray*}
\end{proof}

\subsection{Sampling low-low edges}
\defproc{sample-LL-edge-IS}{Sample-L-L-Edge-IS}

Algorithm \refalg{sample-LL-edge-IS} (\proc*{sample-LL-edge-IS}) samples ``low-low'' edges. Formally, for every edge $uv \in E$, the probability that the output is ``$uv$'' (undirected) is in the range $\frac{e^{\pm \eps}}{m} \Pr[T_u = 0] \Pr[T_v = 0] \pm r$, where $T_u$ and $T_v$ are indicators to accept by the high-degree test (\proc{test-high-degree-IS}) whose error parameter is $r$. When no edge is returned, the algorithm explicitly declares failure by rejecting.

\begin{proc-algo}{sample-LL-edge-IS}{G,n,\eps;\tilde{m}, r}
    \alginput{$\tilde{m} \ge \max\{36, e^{-1/10} m\}$}
    \algoutput{For every edge $e \in E$, the probability to return $e$ is $\frac{e^{\pm \eps}}{200 \tilde{m} \ln (\paperZcoefZestZindZinv/\eps)} \Pr\left[T_u = T_v = 0\right] \pm r$}
    \algcomplexity{$O(\log (n/\eps) + \log(1/r))$ (worst-case)}
    \begin{code}
        \algitem Let $\be \gets \proc{sample-lone-edge-IS}(G,n;\tilde{m})$.
        \begin{If}{$\be$ is \reject}
            \algitem Return \reject.
        \end{If}
        \algitem Let $\bu\bv = \be$ (in an arbitrary order).
        \algitem Filter by $\proc{estimate-indicator-inv}(\proc{loneliness-event-IS}(G,n;\tilde{m},\bu,\bv), \eps, 1/2)$.
        \algitem Let $T_u \gets \proc{test-high-degree-IS}(G,n,\tilde{m};\bu,r)$.
        \algitem Let $T_v \gets \proc{test-high-degree-IS}(G,n,\tilde{m};\bv,r)$.
        \begin{If}{$T_{\bu} \ne 0$ or $T_{\bv} \ne 0$}
            \algitem Return \reject.
        \end{If}
        \algitem Return $\be$.
    \end{code}
\end{proc-algo}

\begin{lemma}{sample-LL-edge-IS}
    Assume that $\tilde{m} \ge \max\{36, e^{-1/10}m\}$. For every edge $uv \in E$, the probability that $\proc{sample-LL-edge-IS}(G,n,\eps;\tilde{m},r)$ returns $u_0 v_0$ is in the range $\frac{e^{\pm \eps}}{200\tilde{m} \ln (\paperZcoefZestZindZinv/\eps)} \Pr[T_{u_0} = 0 \wedge T_{v_0} = 0] \pm r$. Also, the query complexity is $O(\log (n/\eps) + \log (1/r))$ worst-case.
\end{lemma}
\begin{proof}
    For query complexity, observe that there is a single call to \proc{sample-lone-edge-IS} at the cost of $O(\log n)$, a single call to \proc{estimate-indicator-inv} with $\rho=1/2=\Omega(1)$ at the cost of $O(\log (1/\eps))$ queries worst-case (Lemma \reflemma{estimate-indicator-inv}), and two calls to \proc{test-high-degree-IS}, at the cost of $O(\log (1/r))$ queries worst-case.

    Let $u_0 v_0 \in E$. If $\deg(u_0) > 2 k$, then the probability to return the edge $u_0 v_0$ is bounded by $\Pr[T_{u_0} = 0] \le r$ (Lemma \reflemma{test-high-degree-IS}). This applies also when $\deg(v_0) > 2 k$. We proceed assuming that $\deg(u_0)$ and $\deg(v_0)$ are both bounded by $2 k = 2\sqrt{\tilde{m}}$.

    By Lemma \reflemma{sample-lone-edge-IS}, the probability that \proc{sample-lone-edge-IS} returns $u_0 v_0$ is $\mathcal L_{u_0,v_0}(\tilde{m}) / 100 \tilde{m}$.

    By Lemma \reflemma{lbnd-loneliness}, $\mathcal L_{u_0,v_0}(\tilde{m}) \ge \frac{1}{2}$, and therefore, the probability to pass the filter is in the range $e^{\pm \eps} / 2\mathcal L_{u_0,v_0}(\tilde{m}) \ln (\paperZcoefZestZindZinv/\eps)$.

    The probability to return the edge $e_0 = u_0 v_0$ is:
    \begin{eqnarray*}
        && \Pr\left[\be = e_0\right] \cdot \E\left[\text{pass filter}\cond \be = e_0\right] \cdot \Pr\left[T_{u_0} = T_{v_0} = 0\right] \\
        &=& \frac{\mathcal L_{u_0,v_0}(\tilde{m})}{100 \tilde{m}} \cdot \frac{e^{\pm \eps}}{2\mathcal L_{u_0,v_0}(\tilde{m}) \ln (\paperZcoefZestZindZinv/\eps)} \cdot \Pr\left[T_u = T_v = 0\right] \\
        &=& \frac{e^{\pm \eps}}{200 \tilde{m} \ln (\paperZcoefZestZindZinv/\eps)} \cdot \Pr\left[T_u = T_v = 0\right]
    \end{eqnarray*}
\end{proof}

\subsection{Sampling high-low and high-high edges}
\defproc{sample-H-edge-IS}{Sample-H-Edge-IS}

In Algorithm \refalg{sample-H-edge-IS}, implementing \proc*{sample-H-edge-IS}, we sample a star vertex and then sample a neighbor. After normalizing the starness factor and the neighborhood factor, we test the degree of each vertex. If the star vertex degree is considered as low then we reject. Otherwise, we return the edge with probability $1$ if the neighbor's degree is considered low and with probability $1/2$ otherwise. This eliminates the double-counting where both vertices are considered as high-degree.

\begin{proc-algo}{sample-H-edge-IS}{G,n,\eps;\tilde{m},r}
    \alginput{$\tilde{m} \ge \max\{36, e^{-1/10} m\}$}
    \algoutput{For every edge $uv\in E$, the probability to return $uv$ is $\frac{e^{\pm \eps}}{300 \tilde{n}^*\sqrt{\tilde{m}} \ln (\paperZcoefZestZindZinv/\eps)}\Pr\left[T_u = 1 \vee T_v = 1\right] \pm 2r$}
    \algcomplexity{$O(\log n \log (1/\eps))$ (worst-case)}
    \begin{code}
        \algitem Let $k \gets \sqrt{\tilde{m}}$.
        \algitem Let $\bu \gets \proc{sample-star-vertex-IS}(G,n,\eps;\tilde{m})$.
        \begin{If}{$\bu$ is \reject}
            \algitem Return \reject.
        \end{If}
        \algitem Let $\bv \gets \proc{sample-neighbor-IS}(G,n,\tilde{m};\bu)$.
        \begin{If}{$\bv$ is \reject}
            \algitem Return \reject.
        \end{If}
        \algitem Let $\bT_{\bu} \gets \proc{test-high-degree-IS}(G,n,\tilde{m};\bu, r)$.
        \begin{If}{$\bT_{\bu} = 0$}
            \algitem Return \reject.
        \end{If}
        \algitem Let $\bT_{\bv} \gets \proc{test-high-degree-IS}(G,n,\tilde{m}; \bv, r)$.
        \begin{If}{$\bT_{\bv} = 1$}
            \algitem Return \reject with probability $1/2$.
        \end{If}
        \begin{If}{$\tilde{n}^* = n$}
            \algitem Filter by $\proc{estimate-indicator-inv}(\proc{starness-event-IS}(G,n;\tilde{m},\bu), \eps, 1/375)$.
        \end{If}
        \begin{Else}
            \algitem Filter by $\proc{estimate-indicator-inv}(\mathcal A, \eps, 1/750)$ with the following procedure $\mathcal A$:
            \begin{Codeblock*}
                \algitem Let $\bb_1 \gets \proc{starness-event-IS}(G,n;\tilde{m},\bu)$.
                \algitem Let $\bb_2 \gets \proc{neighborhood-event-IS}(G,n;\tilde{m},\bu,\bv)$.
                \algitem Accept if and only if both $\bb_1$ and $\bb_2$ are ``\accept''.
            \end{Codeblock*}
        \end{Else}
        \algitem Return $\bu\bv$.
    \end{code}
\end{proc-algo}

\begin{lemma}{sample-H-edge-IS-directed}
    Assume that $\tilde{m} \ge \max\{36, e^{-1/10} m\}$, and consider the directed edge $u_0 \to v_0$. The probability of Algorithm \refalg{sample-H-edge-IS} (\proc*{sample-H-edge-IS}) to sample the directed edge $u_0 \to v_0$ (through $\bu = u_0$ and $\bv = v_0$) is in the range
    \[\frac{e^{\pm \eps}}{3750\tilde{n}^* \sqrt{\tilde{m}} \ln (\paperZcoefZestZindZinv/\eps)}\Pr\left[\bT_{u_0} = 1\right]\left(\Pr\left[\bT_{v_0} = 0\right] + \frac{1}{2}\Pr\left[\bT_{v_0} = 1\right]\right) \pm r\]
\end{lemma}
\begin{proof}
    For shortening, let $\alpha_{u_0,v_0} = \Pr\left[\bT_{u_0} = 1\right]\left(\Pr\left[\bT_{v_0} = 0\right] + \frac{1}{2}\Pr\left[\bT_{v_0} = 1\right]\right)$.
    
    If $\deg(u_0) \le k$, then the probability to pass the $(\bT_{u_0} = 1)$-test is bounded by $r$. Since $\frac{e^{\eps}}{3750\tilde{n}^* \sqrt{\tilde{m}} \ln (\paperZcoefZestZindZinv/\eps)} < 1$, the probability to sample the directed edge $u_0 \to v_0$ is indeed bounded by $\frac{e^{\pm \eps}}{3750\tilde{n}^* \sqrt{\tilde{m}} \ln (\paperZcoefZestZindZinv/\eps)}\alpha_{u_0,v_0} \pm r$.
    
    If $\deg(u_0) > k$, then:
    \begin{itemize}
        \item The probability of $\bu=u_0$ is $\mathcal S_{u_0}(\tilde{m})/10\sqrt{\tilde{m}}$ (Lemma \reflemma{sample-star-vertex-IS}).
        \item The probability of $\bv=v_0$ is $\frac{1}{n}$ if $\tilde{n}^* = n$ and $\mathcal N_{u_0,v_0}(\tilde{m})/\tilde{m}$ otherwise (Lemma \reflemma{sample-neighbor-IS}).
        \item The probability to pass the test for $\bT_{u_0}$ is $\Pr\left[\bT_u = 1\right]$.
        \item The probability to pass the test for $\bT_{v_0}$ is $\Pr\left[\bT_v = 0\right] + \frac{1}{2}\Pr\left[\bT_v = 1\right]$.
        \item If $\tilde{n}^* = n$: the probability to pass the last filter is $\frac{e^{\pm \eps}}{375 \ln (\paperZcoefZestZindZinv/\eps) \mathcal S_{u_0}(\tilde{m})}$ (Lemma \reflemma{estimate-indicator-inv}, Lemma \reflemma{lbnd-starness}).
        \item If $\tilde{n}^* \ne n$: the probability to pass the last filter is $\frac{e^{\pm \eps}}{750 \ln (\paperZcoefZestZindZinv/\eps) \mathcal S_{u_0}(\tilde{m}) \mathcal N_{u_0,v_0}(\tilde{m})}$(Lemma \reflemma{estimate-indicator-inv}, Lemma \reflemma{lbnd-starness}, Lemma \reflemma{lbnd-neighborhood}).
    \end{itemize}

    If $\tilde{n}^* = n$, then the probability to return the edge $\bu = u_0$ and $\bv = v_0$ is:
    \[  \frac{S_{u_0}(\tilde{m})}{10\sqrt{\tilde{m}}} \cdot \frac{1}{n} \cdot \alpha_{u_0,v_0} \cdot \frac{e^{\pm \eps}}{375 \ln (\paperZcoefZestZindZinv/\eps) \mathcal S_{u_0}(\tilde{m})}
        = \frac{e^{\pm \eps} \alpha_{u_0,v_0}}{3750 n\sqrt{\tilde{m}} \ln (\paperZcoefZestZindZinv/\eps)}
        = \frac{e^{\pm \eps} \alpha_{u_0,v_0}}{3750\tilde{n}^*\sqrt{\tilde{m}} \ln (\paperZcoefZestZindZinv/\eps)}
    \]
    
    If $\tilde{n}^* = 2\tilde{m}$, then the probability to return the edge $u_0 = \bu$ and $v_0 = \bv$ is:
    \[  \frac{S_{u_0}(\tilde{m})}{10\sqrt{\tilde{m}}} \cdot \frac{\mathcal N_{u_0,v_0}(\tilde{m})}{\tilde{m}} \cdot \alpha_{u_0,v_0} \cdot \frac{e^{\pm \eps}}{750 \ln (\paperZcoefZestZindZinv/\eps) \mathcal S_{u_0}(\tilde{m}) \mathcal N_{u_0,v_0}(\tilde{m}) }
        = \frac{e^{\pm \eps} \alpha_{u_0,v_0}}{7500\tilde{m}\sqrt{\tilde{m}} \ln (\paperZcoefZestZindZinv/\eps)}
        = \frac{e^{\pm \eps} \alpha_{u_0,v_0}}{3750\tilde{n}^*\sqrt{\tilde{m}} \ln (\paperZcoefZestZindZinv/\eps)}
    \]    
\end{proof}

\begin{lemma}{sample-H-edge-IS}
    Assume that $\tilde{m} \ge \max\{36, e^{-1/10}m\}$. For every edge $uv \in E$, the probability that $\proc{sample-H-edge-IS}(G,n,\eps;\tilde{m},r)$ returns $u_0 v_0$ is in the range $\frac{e^{\pm \eps}}{3750 \tilde{n}^* \sqrt{\tilde{m}} \ln (\paperZcoefZestZindZinv/\eps)} \Pr[\bT_{u_0} \ge 1 \wedge \bT_{v_0} \ge 1] \pm 2r$. Also, regardless of $\tilde{m}$, the query complexity is $O(\log n \log (1/\eps) + \log (1/r))$ worst-case.
\end{lemma}
\begin{proof}
    For expected complexity assuming that $\tilde{m} \ge \max\{36,e^{-1/10}m\}$, observe that we have:
    \begin{itemize}
        \item A single \proc{sample-star-vertex-IS} call, at the cost of $O(\log n)$ queries (worst-case, Lemma \reflemma{sample-star-vertex-IS}).
        \item A single \proc{sample-neighbor-IS} call, at the cost of $O(\log n)$ queries (worst-case, Lemma \reflemma{sample-neighbor-IS}).
        \item Two calls to \proc{test-high-degree-IS} at the cost of $O(\log (1/r))$ queries (worst-case, Lemma \reflemma{test-high-degree-IS}).
        \item A single call to \proc{estimate-indicator-inv} with parameter $\rho=1/120=\Omega(1)$, at the cost of $O(\log (1/\eps))$ calls to \proc{starness-event-IS} and \proc{neighborhood-event-IS} (worst-case, Lemma \reflemma{estimate-indicator-inv}), each costing $O(\log n)$ (worst-case, Observation \refobs{starness-event-IS}, Observation \refobs{neighborhood-event-IS}).
    \end{itemize}
    
    For shortening, let $\alpha_{u_0,v_0} = \Pr\left[\bT_{u_0} = 0\right]\left(\Pr\left[\bT_{v_0} = 0\right] + \frac{1}{2}\Pr\left[\bT_{v_0} = 1\right]\right)$. Observe that $\alpha_{u_0,v_0} + \alpha_{v_0,u_0} = \Pr\left[\bT_{u_0} = 1 \vee \bT_{v_0} = 1\right]$, since $\bT_{u_0}$ and $\bT_{v_0}$ are independent.
    
    For the undirected edge, we consider both directions.
    \begin{eqnarray*}
        \Pr[\text{sample $u_0 v_0$}]
        &=& \Pr[u_0 \to v_0] + \Pr[v_0 \to u_0] \\
        \text{[Lemma \reflemma{sample-H-edge-IS-directed}]} &=& \left(\frac{e^{\pm \eps} \alpha_{u_0,v_0}}{3750\tilde{n}^*\sqrt{\tilde{m}}} \pm r\right) + \left(\frac{e^{\pm \eps} \alpha_{v_0,u_0}}{3750\tilde{n}^*\sqrt{\tilde{m}}} \pm r\right) \\
        &=& \frac{e^{\pm \eps}}{3750\tilde{n}^*\sqrt{\tilde{m}}}(\alpha_{u_0,v_0} + \alpha_{v_0,u_0}) \pm 2r
        = \frac{e^{\pm \eps}}{3750\tilde{n}^*\sqrt{\tilde{m}}}\Pr\left[\bT_{u_0} = 1 \vee \bT_{v_0} = 1\right] \pm 2r
    \end{eqnarray*}
\end{proof}

\subsection{Sampling all edges}
\defproc{sample-edge-core-IS}{Sample-Edge-Core-IS}
We combine the category-specific edge sampling procedures by unifying their coefficient.

\begin{proc-algo}{sample-edge-core-IS}{G,n,\eps;\tilde{m}}
    \alginput{$0 < \eps \le 1/3$}
    \alginput{$\tilde{m} \ge \max\{36, e^{-1/10} m\}$}
    \algoutput{For every edge $uv\in E$, the probability to return $uv$ is $\frac{e^{\pm 2 \eps}}{1200 \tilde{n}^* \sqrt{\tilde{m}} \ln (\paperZcoefZestZindZinv/\eps)}$}
    \algcomplexity{$O(\log n \log (1 / \eps))$ worst-case}
    \begin{code}
        \algitem Let $r \gets (\eps^2 / \ln (\paperZcoefZestZindZinv/\eps)) / 15000 n^2$.
        \algitem Toss a three-head coin.
        \begin{Codeblock*}
            \algitem With probability $\frac{\sqrt{\tilde{m}}}{75 \tilde{n}^*}$: Return $\proc{sample-LL-edge-IS}(G,n,\eps;\tilde{m}, r)$.
            \algitem With probability $\frac{1}{4}$: Return $\proc{sample-H-edge-IS}(G,n,\eps;\tilde{m}, r)$.
            \algitem Otherwise: \reject.
        \end{Codeblock*}
    \end{code}
\end{proc-algo}

\begin{lemma}{sample-edge-core-IS}
    For every graph $G = (V,E)$ (over $n$ vertices and $m$ edges), if $\tilde{m} \ge \max\{36, e^{-1/10} m\}$ and $0 < \eps \le 1/3$, then Algorithm \refalg{sample-edge-core-IS} (\proc*{sample-edge-core-IS}) returns every edge $uv \in E$ with probability in the range $e^{\pm 2 \eps} / 15000\tilde{m} R \ln (\paperZcoefZestZindZinv/\eps)$ for $R = \min\{\sqrt{\tilde{m}}, n / \sqrt{\tilde{m}}\}$, and otherwise rejects. The query complexity is $O(\log n \log (1/\eps))$ worst-case.
\end{lemma}
\begin{proof}
    For complexity, using $\log (1/r) = O(\log(n/\eps))$:
    \begin{itemize}
        \item The cost of \proc{sample-LL-edge-IS} is $O(\log (n/\eps))$ (Lemma \reflemma{sample-LL-edge-IS}).
        \item The cost of \proc{sample-H-edge-IS} is $O(\log n \log (1/\eps) + \log (n/\eps))$ (Lemma \reflemma{sample-H-edge-IS}).
    \end{itemize}

    Note that $\sqrt{\tilde{m}} / \tilde{n}^* = \sqrt{\tilde{m}} / \min\{ n, 2\tilde{m} \} = \max\{ \sqrt{\tilde{m}}/n, 1/2\sqrt{m} \} \le \max\{ \sqrt{n^2/2}/n, 1/2 \} \le 1$.
    
    The probability to return every individual edge $uv$ is:
    \begin{eqnarray*}
        \Pr\left[\text{sample $uv$}\right]
        &\!\!\!=\!\!\!& \frac{\sqrt{\tilde{m}}}{75\tilde{n}^*} \cdot \left(\frac{e^{\pm \eps}}{200\tilde{m} \ln (\paperZcoefZestZindZinv/\eps)} \Pr\left[T_u = T_v = 0\right] \pm r\right) + \cdots\\&& \frac{1}{4} \cdot \left(\frac{e^{\pm \eps}}{3750 \tilde{n}^* \sqrt{\tilde{m}} \ln (\paperZcoefZestZindZinv/\eps)} \Pr\left[T_u \ge 1 \vee T_v \ge 1\right] \pm 2r\right) \\
        &\!\!\!=\!\!\!& \frac{e^{\pm \eps}}{15000\tilde{n}^* \sqrt{\tilde{m}} \ln (\paperZcoefZestZindZinv/\eps)} \left(\Pr\left[T_u  = T_v  = 0\right] + \Pr\left[T_u \ge 1 \vee T_v \ge 1\right]\right) \pm \left(\frac{1}{75} + 2 \cdot \frac{1}{4}\right)r \\
        &\!\!\!=\!\!\!& \frac{e^{\pm \eps}}{15000\tilde{n}^* \sqrt{\tilde{m}} \ln (\paperZcoefZestZindZinv/\eps)} \pm r
    \end{eqnarray*}

    Observe that for $\eps \le 1/3$:
    \begin{eqnarray*}
        n^2 &\ge& \tilde{n}^* \cdot \sqrt{\tilde{m}} \\
        \eps^2 &\le& \min\{e^{2\eps} - e^{\eps}, e^{-\eps} - e^{-2\eps}\} \\
        r &\le& \frac{\min\{e^{2\eps} - e^\eps, e^{-\eps} - e^{-2\eps}\}}{15000 \tilde{n}^* \sqrt{\tilde{m}} \ln (\paperZcoefZestZindZinv/\eps)}
    \end{eqnarray*}
    
    Therefore,
    \[  \Pr\left[\text{sample $uv$}\right]
        = \frac{e^{\pm \eps}}{15000\tilde{n}^* \sqrt{\tilde{m}} \ln (\paperZcoefZestZindZinv/\eps)} \pm \frac{\min\{e^{2\eps} - e^{\eps}, e^{-\eps} - e^{-2\eps}\}}{15000\tilde{n}^* \tilde{m} \ln (\paperZcoefZestZindZinv/\eps)}
        = \frac{e^{\pm 2\eps}}{15000\tilde{n}^* \sqrt{\tilde{m}} \ln (\paperZcoefZestZindZinv/\eps)}
    \]
\end{proof}

\defproc{sample-edge-amplified-IS}{Sample-Edge-Amplified-IS}
The core sampler is a $(\lambda,\eps)$-uniform sampler, but $\lambda$ can be very small. Algorithm \refalg{sample-edge-amplified-IS} (\proc*{sample-edge-amplified-IS}) amplifies the success probability to $3/4$ by making $O(R \log (1/\eps))$ sample tries, for $R = \min\{\sqrt{m}, n / \sqrt{m}\}$, and choosing the first successful one. This algorithm requires an advice $\tilde{m}$, and its correctness and expected query complexity statements only hold if $\tilde{m} \in e^{\pm 1/10} m$.

\begin{proc-algo}{sample-edge-amplified-IS}{G, n, \tilde{m}, \eps}
    \alginput{$\tilde{m} \ge \max\{36, e^{-1/10} m\}$}
    \algoutput{For every edge $e \in E$, the probability to return $e$ is $e^{\pm \eps} \lambda(G,\tilde{m})$}
    \algoutputphantom{}{($\lambda(G,\tilde{m}) \ge 3/4$ if $\tilde{m} \ge e^{-1/10} m$)}
    \algcomplexity{$O(\min\{\sqrt{m}, n / \sqrt{m}\} \cdot \log n \log (1/\eps) \log (n/\eps))$ expected, if $\tilde{m} \ge e^{-1/10} m$}
    \algcomplexity{$O(m\sqrt{n}\log n \log (1/\eps))$ worst-case}
    \begin{code}
        \algitem Let $\tilde{R} \gets \min\{\sqrt{\tilde{m}}, n / \sqrt{\tilde{m}}\}$.
        \begin{For}{$\ceil{10^6 \tilde{R} \ln (\paperZcoefZestZindZinv/\eps)}$ times}
            \algitem Let $\be \gets \proc{sample-edge-core-IS}(G,n,\eps/5;\tilde{m})$.
            \begin{If}{$\be$ is not \reject}
                \algitem Return $\be$.
            \end{If}
        \end{For}
        \algitem Return \reject.
    \end{code}
\end{proc-algo}

\begin{lemma}{sample-edge-amplified-IS}
    For every graph $G = (V,E)$ (over $n$ vertices and $m$ edges), if $\tilde{m} \ge \max\{36, e^{-1/10} m\}$ and $0 < \eps \le 1/3$, then Algorithm \refalg{sample-edge-amplified-IS} (\proc*{sample-edge-amplified-IS}) returns every edge $uv \in E$ with probability in the range $e^{\pm \eps} \lambda(G,\tilde{m})$, where $\lambda(G,\tilde{m}) \ge 3/4$, and otherwise rejects. The query complexity is $O(R \log n \log^2 (1/\eps))$ if $\tilde{m} \in e^{\pm 1/10} m$ and otherwise $O(\sqrt{n} \log n \log^2 (1/\eps))$.
\end{lemma}
\begin{proof}
    Let $R = \min\{\sqrt{m}, n / \sqrt{m}\}$. The success probability of an iteration is at least:
    \[  \frac{e^{-(2/5)\eps} m}{15000 \tilde{n^*} \sqrt{\tilde{m}} \ln (\paperZcoefZestZindZinv/\eps)}
        \ge \frac{m}{30000e^{2/15} \cdot \tilde{R} \cdot e^{1/10}m \ln (\paperZcoefZestZindZinv/\eps)}
        \ge \frac{1}{40000 \tilde{R} \ln (\paperZcoefZestZindZinv/\eps)}
    \]
    
    Therefore, the success probability of the loop is:
    \[  \lambda(G,\tilde{m})
        \ge 1 - \left(1 - \frac{1}{40000 \tilde{R} \ln (\paperZcoefZestZindZinv/\eps)}\right)^{10^6 \tilde{R} \ln (\paperZcoefZestZindZinv/\eps)}
        \ge 1 - e^{-25}
        \ge 3/4
    \]

    For every edge $e \in E$, the probability to sample it is in the range $e^{\pm 2 (\eps/5)} \lambda(G, \tilde{m})$.

    The query complexity is $O(\tilde{R} \log (1/\eps))$ rounds, $O(\log n \log (1/\eps))$ per round (Lemma \reflemma{sample-edge-core-IS}). If $\tilde{m} \ge 36$ and $\tilde{m} \in e^{\pm 1/10} m$, then $\tilde{R} = O(R)$, and otherwise the worst case is $\tilde{R} = O(\sqrt{n})$.
\end{proof}

\defproc{sample-edge-IS}{Sample-Edge-IS}
The entry-point sampling algorithm, \proc*{sample-edge-IS} (Algorithm \refalg{sample-edge-IS}), estimates $\tilde{m} \in e^{\pm 1/10} m$ with success probability $1 - r$, where $r = O(\eps^2) \cap O(1/n^3)$.

\begin{proc-algo}{sample-edge-IS}{G,n,\eps}
    \begin{code}
        \begin{If}{$\eps > 1/3$}
            \algitem Set $\eps \gets 1/3$.
        \end{If}
        \algitem Let $r \gets \min\{\eps^2/6, 1/n^2 \log n \}$.
        \algitem Compute $\tilde{\bm} \gets e^{\pm 1/10} m$ with probability $1 - r$.
        \begin{If}{$\tilde{\bm} < 36$}
            \algitem Return $\proc{sample-edge-bruteforce-IS}(G)$.
        \end{If}
        \algitem Return $\proc{sample-edge-amplified-IS}(G, n, \tilde{\bm}, \eps/5)$.
    \end{code}
\end{proc-algo}

\begin{lemma}{sample-edge-IS}
    For every graph $G = (V,E)$ (over $n$ vertices and $m$ edges). Algorithm \refalg{sample-edge-IS} (\proc{sample-edge-IS}) returns every edge $uv \in E$ with probability in the range $e^{\pm 4\eps} / m$, and otherwise rejects. The expected query complexity is $O(R \poly(\log n) \log^2 (1/\eps))$ for $R = \min\{\sqrt{m}, n/\sqrt{m}\}$.
\end{lemma}
\begin{proof}
    Estimating $\tilde{m}$ in the range $e^{\pm 1/10} m$ with probability $2/3$ requires $O(R \cdot \poly(\log n))$ queries worst-case (Lemma \reflemma{external:estimate-edges-IS}). To increase this probability to $1 - r$, we take the median of $O(\log r^{-1}) = O(\log (n/\eps))$ independent runs at the cost of $O(R \poly(\log n) \log (n/\eps))$ queries.
    
    The contribution of the brute-force sampler in the case where $\tilde{\bm}$ is wrongly small is $\Pr\left[\tilde{\bm} < 36\right] \cdot O(1 + m \log n) = O(r m \log n) = O(1)$ (last transition is correct since $m \le n^2$ and $r \le 1/n^2 \log n$).

    The amplified sample costs $O(R \cdot \log n \log^2 (1/\eps))$ if $\tilde{m} \in e^{\pm 1/10} m$ and $O(\sqrt{n} \log n \log^2 (1/\eps))$ if $\tilde{m} \notin e^{-1/10} m$, which happens with probability at most $r \le 1/n^2 \log n$. Therefore, the expected query complexity of the sampling logic is $O(R \log n \log^2 (1/\eps))$ (in addition to the estimation of $\tilde{m}$).
    
    Let $\lambda'(G) = \frac{\sum_{a \in e^{\pm 1/10} m} \Pr\left[\tilde{m} = a\right]}{\Pr\left[\tilde{m} \in e^{\pm 1/10} m\right]} \ge \frac{3}{4}$ (last transition: Lemma \reflemma{sample-edge-amplified-IS}).

    By Lemma \reflemma{sample-edge-amplified-IS} and the union bound, the probability to sample an individual edge $e \in E$ is at least $e^{-\eps/5} \lambda'(G)$ and at most $e^{\eps/5}\lambda'(G) + r$. For $0 < \eps \le 1/3$, the ratio between the maximum and the minimum probability is at most:
    \begin{eqnarray*}
        e^{(2/5)\eps} + \frac{e^{\eps/5}}{\lambda'(G)} r
        \le e^{(2/5)\eps} + \frac{4}{3} e^{\eps/5} \cdot \frac{1}{6} \eps^2
        \le e^{\eps/2}
    \end{eqnarray*}

    Therefore, the output is $\eps$-uniform when conditioned on success.
\end{proof}

\section{Lower bounds}
\label{sec:lbnd}

We describe a paradigm for showing lower bounds for sampling algorithms. Consider two graphs $G_1 = (V,E_1)$ and $G_2 = (V,E_2)$ for which $E_1 \subseteq E_2$. For every relabeling $\pi : V \to V$, let $E^\pi_b$ be the set of edges in $\pi(G_b)$ ($b \in \{1,2\}$).

We formally define the algorithm behavior over a random input as a distribution.
\begin{definition}{alg-distribution-of-runs}
    Let $\mathcal A$ be a query-making algorithm, and let $\mathcal G$ be a distribution over inputs. We define the \emph{distribution of runs} of $\mathcal A$ given an input drawn from $\mathcal G$ as the distribution for which, for every query-answer sequence $\sigma$, \[\Pr_{\mathcal A(\mathcal G)}\left[\sigma\right] = \Pr_{\bG \sim \mathcal G}\left[\text{$\mathcal A$ executes the query-answer sequence $\sigma$ when given an input $\bG$}\right].\]
\end{definition}

Consider the model of double inputs: the algorithm input consists of two graphs, and every graph query is performed on both of them. If a single-input algorithm cannot distinguish between $\pi(G_1)$ and $\pi(G_2)$ using a specific sequence of queries, then it cannot distinguish between them even if $\pi(G_1)$ is given in addition to the actual input $G \in \{ \pi(G_1), \pi(G_2) \}$. This is stated in the following observation:
\begin{observation}{single-input-to-double-input-analysis}
    Let $G_1$ and $G_2$ be two graphs over the same vertex set $V$, and let $\bpi : V \to V$ be a random relabeling of the vertices. Assume that for every single-graph algorithm $\mathcal A$ that makes $q$ queries in expectation when given an input of the form $\bpi(G_1)$, $\dtv(\mathcal A(\bpi(G_1)), \mathcal A(\bpi(G_2))) \le d(q)$. In this setting, for every double-graph algorithm $\mathcal A'$ that makes $q$ queries in expectation when given an input of the form $(\bpi(G_1), \bpi(G_1))$, $\dtv(\mathcal A'(\bpi(G_1), \bpi(G_1)), \mathcal A'(\bpi(G_1), \bpi(G_2))) \le d(q)$.
\end{observation}

\begin{lemma}{sampler-to-double-reduction}
    Let $G_1$ and $G_2$ be two graphs over the same vertex set $V$, and let $\bpi : V \to V$ be a random relabeling of the vertices. Assume that algorithm $\mathcal A$ that is allowed to make queries, in a query model that allows determining the existence of a given edge\footnote{Such as every model that allows IS queries.}, is a $(\lambda,\eps)$-uniform sampler that makes $q$ queries in expectation when given an input of the form $\bpi(G_1)$. In this setting, there exists a double-input algorithm $\mathcal A'$ that makes $q+1$ queries in expectation when given an input of the form $\bpi(G_1)$, for which \[
        \dtv(\mathcal A(\bpi(G_1)), \mathcal A(\bpi(G_2)))
        = \dtv(\mathcal A'(\bpi(G_1), \bpi(G_1)), \mathcal A'(\bpi(G_1), \bpi(G_2)))
        \ge \lambda e^{-\eps} \frac{\abs{E_2 \setminus E_1}}{\abs{E_2}}.\]
\end{lemma}
\begin{proof}
    We define $\mathcal A'(H_1,H_2)$ as the algorithm that runs $\mathcal A$ as-is on the graph $H_2$, and then makes an additional IS query to determine whether the result edge belongs to $H_1$. Since $\mathcal A$ is a $(\lambda,\eps)$-uniform sampler, the probability to sample an edge in $E_2 \setminus E_1$ from $H_1$ is zero, and the probability to sample such an edge from $H_2$ is at least $\lambda \cdot e^{-\eps} \cdot \abs{E_2 \setminus E_1} \cdot \frac{1}{\abs{E_2}}$. That is, the additional query enforces a $\lambda e^{-\eps} \frac{\abs{E_2 \setminus E_1}}{\abs{E_2}}$-difference in the behavior of the algorithm when given inputs of the form $(\bpi(G_1), \bpi(G_1))$ and inputs of the form $(\bpi(G_1), \bpi(G_2))$.

    Observe that the total-variation distance of the algorithm behaviors is the same, since the last query of $\mathcal A'$ is different only if the execution of $\mathcal A$ has already forked into paths resulting in different outputs.
\end{proof}

\begin{lemma}[Lower bound mechanism]{sampling-lower-bound-mechanism}
    Let $n_0 \ge 1$, $m_0 \ge 1$, $0 < \rho < 1$, $0 < \eps_0 < 1$ and some $R(n,m) > 0$. Assume that there exists some globally fixed constant $C$ such that for every $n \ge n_0$ and every $m_0 \le m \le \rho n^2$ there exist two graphs $G_1$ and $G_2$ on the same vertex set $V$ for which:
    \begin{itemize}
        \item Considering the edge sets, $E_1 \subseteq E_2$.
        \item $\abs{E_1} \le m$.
        \item $\abs{E_2} - \abs{E_1} \ge \eps_0 m$.
        \item $\abs{E_2} - \abs{E_1} \le m$.
        \item Let $\bpi : V \to V$ be a uniformly chosen relabeling of the vertices. For every algorithm that makes $q$ queries in expectation when given an input of the form $\bpi(G_1)$, $\dtv(\mathcal A(\bpi(G_1)), \mathcal A(\bpi(G_2))) \le \frac{C}{R(n,m)} \cdot q$.
    \end{itemize}
    In this setting, every $(\lambda,\eps_0)$-uniform edge sampler must make $\frac{1}{2C}\lambda e^{-\eps_0} \eps_0 R(n,m)$ queries in expectation when given an input of the form $\bpi(G_1)$.
\end{lemma}
\begin{proof}
    By Lemma \reflemma{sampler-to-double-reduction}, if $\mathcal A$ is a $(\lambda,\eps_0)$-uniform sampler, then for an appropriate double-input algorithm $\mathcal A'$ that makes one more query:
    \[  \dtv(\mathcal A'(\bpi(G_1), \bpi(G_1)), \mathcal A'(\bpi(G_1), \bpi(G_2)))
        \ge \lambda e^{-\eps_0} \eps_0 \frac{\abs{E_2 \setminus E_1}}{\abs{E_2}}
        \ge \lambda e^{-\eps_0} \frac{\eps_0 m}{2m}
        = \frac{1}{2}\lambda e^{-\eps_0} \eps_0.
    \]
    Since the total variation distance in the behavior of the modified algorithm is bounded by the total variation distance in the behavior of the sampling algorithm,
    \[ \frac{C}{R(n,m)} \cdot q \ge \frac{1}{2}\lambda e^{-\eps_0} \eps_0. \]
    
    This results in $q \ge \frac{1}{2C}\lambda e^{-\eps_0} \eps_0 \cdot R(n,m)$.
\end{proof}

For hybrid algorithms:
\begin{lemma}[Lower bound for hybrid algorithms]{lbnd-sample-edge-hyb}
    For a sufficiently small, fixed $\eps$, every $\eps$-uniform edge-sampling algorithm that uses IS queries and local queries must make at least $\Omega(\min\{\sqrt{m}, \sqrt{n/\sqrt{m}}\})$ queries in expectation.
\end{lemma}
\begin{proof}
    For some globally fixed constants $n_0 \ge 1$, $m_0 \ge 1$ and $0 < \eps_0 < 1$, let $n \ge n_0$, $m \ge m_0$ and $0 < \eps \le \eps_0$. Also, let $R(n,m) = \min\{\sqrt{m}, \sqrt{n/\eps\sqrt{m}}\}$.

    By \cite{ahl26}, for a globally fixed constant $C$ (independent of $n$, $m$ and $\eps$), there exist two graphs named ``$G_{n,m/n,0}$'' (denoted here by $G_1$) and ``$G_{n,m/n,\eps}$'' (denoted here by $G_2$) so that:
    \begin{itemize}
        \item Considering the edge sets, $E_1 \subseteq E_2$.
        \item $\abs{E_1} \le m$.
        \item $\abs{E_2} - \abs{E_1} \ge \eps m$.
        \item $\abs{E_2} - \abs{E_1} \le m$ (stated in \cite{ahl26} as $\Theta(\eps m)$).
        \item For every algorithm that makes $q$ queries in expectation, $\dtv(\mathcal A(\bpi(G_1)), \mathcal A(\bpi(G_2))) \!\le\! \frac{C}{R(n,m)} \!\cdot\! q$.
    \end{itemize}
    This matches the constraints of Lemma \reflemma{sampling-lower-bound-mechanism} for $R(n,m) = \min\{\sqrt{m}, \sqrt{n/\sqrt{m}}\}$.
\end{proof}

For the IS algorithm, the lower bound of \cite{clw20} for degree estimation is $\tilde{\Omega}(\min\{\sqrt{m}, \frac{n}{\sqrt{m}}\})$, and the edge containment constraint ($E_1 \subseteq E_2$) is presented in the text but not explicitly stated. Below we attach an alternative proof, based on the construction ideas of \cite{ahl26}, for an $\Omega(\min\{\sqrt{m}, \frac{n}{\sqrt{m}}\})$ lower bound for a fixed $\eps$.

\begin{lemma}{lower-bound-IS-construction}
    For every $n \ge 16$ and $9 \le m \le n^2 / 16$, there exist two graphs $G_{n,m}$ and $H_{n,m}$ over $\{1,\ldots,n\}$ for which:
    \begin{itemize}
        \item $G_{n,m}$ is a subgraph of $H_{n,m}$.
        \item $G_{n,m}$ contains at most $m$ edges.
        \item $H_{n,m}$ contains at least $\frac{1}{2} m$ additional edges.
        \item $H_{n,m}$ contains at most $m$ additional edges.
        \item There exists a globally fixed constant $C$ so that, for every algorithm that only uses IS queries and makes $q$ queries in expectation, $\dtv(\mathcal A(\bpi(G_{n,m})), \mathcal A(\bpi(H_{n,m}))) \le \frac{C}{R(n,m)} \cdot q$, for $R(n,m) = \min\{\sqrt{m}, n/ \sqrt{m}\}$.
    \end{itemize}
\end{lemma}
We prove each sub-statement of Lemma \reflemma{lower-bound-IS-construction} in Lemmas \reflemma{IS-lbnd-construction-well-defined}, \reflemma{IS-lbnd-construction-G-edges}, \reflemma{IS-lbnd-construction-H-G-edge-diff}, \reflemma{IS-lower-bound-construction-dtv}.

\begin{lemma}[Lower bound for IS algorithms]{lbnd-sample-edge-IS}
    For a sufficiently small, fixed $\eps$, every $\eps$-uniform edge-sampling algorithm that uses IS queries must make at least $\Omega(\min\{\sqrt{m}, n/\sqrt{m}\})$ queries.
\end{lemma}
\begin{proof}
    Just apply Lemma \reflemma{sampling-lower-bound-mechanism} with the construction stated in Lemma \reflemma{lower-bound-IS-construction}.
\end{proof}

\subsection{The IS lower-bound construction}
Our construction is based on the clique+biclique idea of the lower bound presented in \cite{ahl26}. The following analysis is shorter since it refers to a hard-coded choice of $\eps$ and only considers IS queries.

For $n \ge 16$ and $9 \le m \le \frac{1}{16}n^2$, we define $k = \floor{\sqrt{m}}$, $m' = m - \binom{k}{2}$, $h = \ceil{2m'/n}$, $\ell = \ceil{m'/h}$ and two graphs:
\begin{itemize}
    \item $G_{n,m}$ is the union of a $k$-clique and $n-k$ singleton vertices.
    \item $H_{n,m}$ as the union of a $k$-clique, an $(\ell,h)$-biclique and $n - (k + h + \ell)$ singleton vertices.
\end{itemize}

Note that $\ell \approx \frac{1}{2}n$ is the  number of low-degree vertices and $h \approx m/n$ is the number high-degree vertices (note that $h\le\ell)$.

\begin{lemma}{IS-lbnd-construction-well-defined}
    The graphs $G_{n,m}$ and $H_{n,m}$ are well-defined.
\end{lemma}
\begin{proof}
    It suffices to show that the number of singletons is non-negative in the $H$-graph.
    \begin{eqnarray*}
        k + h + \ell
        &\le& \sqrt{m} + \frac{2m'}{n} + 1 + \frac{m'}{\ceil{2m'/n}} + 1 \\
        &\le& \sqrt{m} + \frac{2m}{n} + \frac{1}{2}n + 2
        \le \frac{1}{4}n + \frac{1}{8}n + \frac{1}{2}n + 2
        = \frac{7}{8}n + 2
        \le n
    \end{eqnarray*}
\end{proof}

\begin{lemma}{IS-lbnd-construction-G-edges}
    The number of edges in $G_{n,m}$ is at most $\frac{1}{2}m$.
\end{lemma}
\begin{proof}
    The number of edges in $G_{n,m}$ is
    \[  \binom{k}{2}
        \le \frac{1}{2}k^2
        \le \frac{1}{2}\floor{\sqrt{m}}^2
        \le \frac{1}{2}m
        \qedhere \]
\end{proof}

\begin{lemma}{IS-lbnd-construction-H-G-edge-diff}
    The number of edges in $H_{n,m}$ not belonging to $G_{n,m}$ is at least $\frac{1}{2}m$.
\end{lemma}
\begin{proof}
    The number of additional edges in $H_{n,m}$ is:
    \[h \cdot \ell
        = h \cdot \ceil{m'/h}
        \ge m'
        = m - \binom{k}{2}
        \ge m - \frac{1}{2}m
        = \frac{1}{2}m
        \qedhere \]
\end{proof}

\begin{lemma}{IS-lbnd-construction-H-G-edge-diff-ubnd}
    The number of edges in $H_{n,m}$ not belonging to $G_{n,m}$ is at most $m$.
\end{lemma}
\begin{proof}
    For $m \ge 9$:
    \[ \binom{k}{2}
        = \binom{\floor{\sqrt{m}}}{2}
        \ge \frac{1}{5}m
    \]
    (Minimum is achieved where $m+1$ is a square. Can be validated explicitly for $m=15$ and algebraically for $m \ge 24$ through $\floor{x} \ge x-1$).

    For $m'$:
    \[  m'
        = m - \binom{k}{2}
        \le m - \frac{1}{5}m
        = \frac{4}{5}m
    \]
    
    Since $n \ge 16$, the number of additional edges in $H_{n,m}$ is:
    \[  h \cdot \ell
        \le h \cdot (m'/h + 1)
        = m' + h
        \le m' + (2m'/n + 1)
        = (1 + 2/n) m' + 1
        \le \frac{9}{8} m' + 1
        \le \frac{9}{8} \cdot \frac{4}{5}m  + 1
    \]
    For $m \ge 10$ we directly obtain that $h \cdot \ell \le m$. For $m=9$, we obtain that $h \cdot \ell \le 9.1$, but since it must be an integer, it is bounded by $m$.
\end{proof}

We color the vertices using four colors: $\mathrm{A}$ (singletons), $\mathrm{K}$ (clique), $\mathrm{L}$ (low-degree part, $\ell \approx n/2$ vertices) and $\mathrm{H}$ (high-degree part, $h \approx m/n$ vertices). Observe that the result of an independent-set query only depends on its vertex colors.

To analyze the lower bound, we consider a deterministic adaptive algorithm. This lower bound can be extended to probabilistic algorithms through Yao's principle \cite{yao77}.

We use a simulation argument similar to \cite{ahl26}. We translate every IS query into a sequence of ``color queries''. In each color query, we reveal the color of a single vertex (of $\bpi(G)$ or $\bpi(H)$). Note that all non-$\mathrm{K}$-colored vertices in $\bpi(H)$ have $\mathrm{A}$-color in $\bpi(G)$. Since additional edges in $H$ must be between an $\mathrm{L}$-color vertex and an $\mathrm{H}$-color vertex, the behavior can differ only when an $\mathrm{A}$-colored vertex in $\bpi(G)$ has an $\mathrm{H}$-color in $\bpi(H)$. This is stated in the following observation.

\begin{observation}{color-simulation}
    Let $\mathcal A'$ be a double-input algorithm. The distance between runs of $\mathcal A'$ when given an input $(\bpi(G), \bpi(G))$ and runs of $\mathcal A'$ when given an input $(\bpi(G), \bpi(H))$, when conditioned on the number of color queries $\bQ$ made when given the input $(\bpi(G), \bpi(G))$, is bounded by $\frac{h}{n - k} \bQ$.
\end{observation}

For a uniformly chosen relabeling $\bpi : V \to V$, we consider a run of a deterministic adaptive algorithm on the input $\bpi(G_{n,m})$, and then examine this run to bound the probability of the algorithm to behave differently on the input $\bpi(H_{n,m})$.

\begin{lemma}{IS-lower-bound-construction-expval-color-queries}
    Assume that we simulate every IS query in the graph $G_{n,m}$ by a sequence of color queries. Given query set $S$, we sequentially query the vertices of $S$ until we obtain two $\mathrm{K}$-vertices, in which case we return ``not independent'', or until we exhaust all vertices, in which case we return ``independent''. In this setting, for $q \le \frac{1}{4}k$, the first $q$ query simulations reveal at most $(4n/k) q$ vertex colors in expectation.
\end{lemma}
\begin{proof}
    For $1 \le i \le q$, assume that (at most) $2(i-1)$ $\mathrm{K}$-color vertices have already been found. The probability of a ``new'' vertex to be $\mathrm{K}$-color is at least $\frac{U_K}{U_V} \ge \frac{k - 2(i-1)}{n}$, where $U_K$ is the number of unseen $\mathrm{K}$-color vertices and $U_V$ is the number of unrevealed vertices, and this probability cannot decrease as we find $\mathrm{A}$-color vertices (since such reveal decrease the denominator while keeping the numerator unchanged). Therefore, the expected number of color queries until finding a new $\mathrm{K}$-color vertex is bounded by $\frac{n}{k - 2(i-1)}$. By the same analysis, the expected number of color queries until finding the second $\mathrm{K}$-color vertex is bounded by $\frac{n}{k - (2(i-1) + 1)}$. Note that this search may terminate even earlier if we exhaust all $S_i$ vertices.

    Combined, the expected number of color queries in the first $r \le q$ queries is bounded by:
    \begin{eqnarray*}
        \sum_{i=1}^q \left(\frac{n}{k - 2(i-1)} + \frac{n}{k - (2(i-1) + 1)}\right)
        \le \sum_{i=1}^q \left(2 \cdot \frac{n}{k - 2q}\right)
        \le 2 \sum_{i=1}^q \frac{n}{k - 2k/4}
        = \frac{4n}{k} \cdot q
    \end{eqnarray*}
\end{proof}

\begin{lemma}{IS-lower-bound-construction-dtv}
    Consider an algorithm that when given an input of the form $\bpi(G_{n,m})$, where $\bpi$ is a uniformly chosen relabeling, makes $\bQ$ IS queries. The total variation distance between the behavior of the algorithm on the input $\bpi(G_{n,m})$ and its behavior on the input $\bpi(H_{n,m})$ is bounded by $64\E[\bQ]/R$ for $R = \min\{\sqrt{m}, n/\sqrt{m}\}$.
\end{lemma}
\begin{proof}
    As observed before (Observation \refobs{color-simulation}), the total-variation distance between the runs is bounded by the product of (1) the ratio between $h$ and the number of singletons in $G_{n,m}$ and (2) the number of color queries.
    
    By Lemma \reflemma{IS-lower-bound-construction-expval-color-queries}, the expected number of color queries is bounded by:
    \[
        \sum_{r=1}^{\floor{k/4}} \Pr\left[\bQ = r\right] \cdot \frac{4n}{k} r + n \cdot \Pr\left[\bQ > k/4\right]
        \le \frac{4n}{k}\sum_{r=1}^{\floor{k/4}} \Pr\left[\bQ = r\right] r + n \cdot \frac{\E\left[\bQ\right]}{k/4}
        \le \frac{8n}{k} \E\left[\bQ\right] \]

    Bounds:
    \begin{eqnarray*}
        \dtv(\mathcal A(\bpi(G_{n,m})), \mathcal A(\bpi(H_{n,m})))
        &\le& \frac{h}{n - k} \cdot \E\left[\text{color queries}\right] \\
        &\le& \frac{\ceil{2m'/n}}{n - \floor{\sqrt{m}}} \cdot \frac{8n}{k}\E\left[\bQ\right] \\
        &\le& \frac{2m/n + 1}{(3/4) n} \cdot \frac{8n}{\sqrt{m}/2}\E[\bQ] \\
        &=& \left(\frac{(128/3)\sqrt{m}}{n} + \frac{64/3}{\sqrt{m}}\right) \E[\bQ] \\
        &\le& 64\max\{\sqrt{m}/n, 1/\sqrt{m}\} \E[\bQ] \\
        &=& \frac{64}{\min\{\sqrt{m}, n/\sqrt{m}\}} \E[\bQ]
    \end{eqnarray*}
\end{proof}

\newpage
\phantomsection
\addcontentsline{toc}{section}{References}

\bibliographystyle{alpha}
\bibliography{main}

\end{document}